\documentclass[english,11pt]{article}
\usepackage[english]{babel}
\usepackage{latexsym,amssymb,amsmath,amsfonts,amscd,epsfig,amsthm,color}
\usepackage{graphicx}
\usepackage{float}
\usepackage[left=1in,right=1in,top=1in,bottom=1in]{geometry}
\usepackage{authblk}
\usepackage{tikz}

\usepackage{hyperref}
\hypersetup{
    bookmarksnumbered=true, 
    unicode=false, 
    pdfstartview={FitH}, 
    pdftitle={Title}, 
    pdfauthor={Authors}, 
    pdfsubject={}, 
    pdfcreator={}, 
    pdfproducer={}, 
    pdfkeywords={}, 
    pdfnewwindow=true, 
    colorlinks=true, 
    linkcolor=blue, 
    citecolor=blue, 
    filecolor=blue, 
    urlcolor=blue 
}

\newcommand{\eq}[1]{\hyperref[eq:#1]{(\ref*{eq:#1})}}
\renewcommand{\sec}[1]{\hyperref[sec:#1]{Section~\ref*{sec:#1}}}
\newcommand{\thm}[1]{\hyperref[thm:#1]{Theorem~\ref*{thm:#1}}}
\newcommand{\lem}[1]{\hyperref[lem:#1]{Lemma~\ref*{lem:#1}}}
\newcommand{\cor}[1]{\hyperref[cor:#1]{Corollary~\ref*{cor:#1}}}
\newcommand{\app}[1]{\hyperref[app:#1]{Appendix~\ref*{app:#1}}}
\newcommand{\tab}[1]{\hyperref[tab:#1]{Table~\ref*{tab:#1}}}
\newcommand{\defin}[1]{\hyperref[def:#1]{Definition~\ref*{def:#1}}}
\newcommand{\fig}[1]{\hyperref[fig:#1]{Figure~\ref*{fig:#1}}}

\newcommand{\thmthm}[2]{\hyperref[thm:#1]{Theorem~\ref*{thm:#1}} and~\hyperref[thm:#2]{\ref*{thm:#2}}}
\newcommand{\lemlem}[2]{\hyperref[lem:#1]{Lemma~\ref*{lem:#1}} and~\hyperref[lem:#2]{\ref*{lem:#2}}}

\newtheorem{theorem}{Theorem}[section]
\newtheorem{lemma}[theorem]{Lemma}

\newtheorem{corollary}[theorem]{Corollary}

\newtheorem{definition}[theorem]{Definition}

\def\ket#1{{\lvert}#1\rangle}
\def\bra#1{{\langle}#1\rvert}
\def\braket#1#2{{{\langle}#1\vert}#2\rangle}
\def\abs#1{\left| #1 \right|}
\def\ceil#1{{\lceil}#1\rceil}

\def\norm#1{\left\| #1 \right\|}

\def\span{\mbox{span}}
\newcommand{\eps}{\varepsilon}
\def\({\left(}
\def\){\right)}
\def\th#1{#1$^\mathrm{th}$}
\def\tO{\widetilde{O}}
\def\overarrow#1{\overset{\rightarrow}{#1}}

\author{Tsuyoshi Ito and Stacey Jeffery\thanks{sjeffery@caltech.edu, Institute for Quantum Information and Matter, California Institute of Technology}}

\date{}

\title{Approximate Span Programs}

\begin{document}

\maketitle

\begin{abstract}
Span programs are a model of computation that have been used to design quantum algorithms, mainly in the query model. It is known that for any decision problem, there exists a span program that leads to an algorithm with optimal quantum query complexity, however finding such an algorithm is generally challenging. 

In this work, we consider new ways of designing quantum algorithms using span programs. 
We show how any span program that decides a problem $f$ can also be used to decide ``property testing'' versions of the function $f$, or more generally, approximate a quantity called the \emph{span program witness size}, which is some property of the input related to $f$. For example, using our techniques, the span program for OR, which can be used to design an optimal algorithm for the  OR function,  can also be used to design optimal algorithms for: threshold functions, in which we want to decide if the Hamming weight of a string is above a threshold, or far below, given the promise that one of these is true; and approximate counting, in which we want to estimate the Hamming weight of the input up to some desired accuracy. We achieve these results by relaxing the requirement that 1-inputs hit some \emph{target} exactly in the span program, which could potentially make design of span programs significantly easier. 

In addition, we give an exposition of span program structure, which increases the general understanding of this important model. One implication of this is alternative algorithms for estimating the witness size when the phase gap of a certain unitary can be lower bounded. We show how to lower bound this phase gap in certain cases.

As an application, we give the first upper bounds in the adjacency query model on the quantum time complexity of estimating the effective resistance between $s$ and $t$, $R_{s,t}(G)$. For this problem we obtain $\tO(\frac{1}{\eps^{3/2}}n\sqrt{R_{s,t}(G)})$, using $O(\log n)$ space. In addition, when $\mu$ is a lower bound on $\lambda_2(G)$, by our phase gap lower bound, we can obtain an upper bound of $\tO\(\frac{1}{\eps}n\sqrt{R_{s,t}(G)/\mu}\)$ for estimating effective resistance, also using $O(\log n)$ space. 
\end{abstract}

\section{Introduction}

Span programs are a model of computation first used to study logspace complexity \cite{KW93}, and more recently, introduced to the study of quantum algorithms in \cite{RS12}. They are of immense theoretical importance, having been used to show that the general adversary bound gives a tight lower bound on the quantum query complexity of any decision problem \cite{Rei09,Rei11}. As a means of designing quantum algorithms, it is known that for any decision problem, there exists a span-program-based algorithm with asymptotically optimal quantum query complexity, but this fact alone gives no indication of how to find such an algorithm. Despite the relative difficulty in designing quantum algorithms this way, there are many applications, including formula evaluation \cite{RS12,Rei11}, a number of algorithms based on the learning graph framework \cite{Bel11}, $st$-connectivity \cite{BR12} and $k$-distinctness \cite{Bel12}. Although generally quantum algorithms designed via span programs can only be analyzed in terms of their query complexity, in some cases their time complexity can also be analyzed, as is the case with the quantum algorithm for $st$-connectivity. In the case of the quantum algorithm for $k$-distinctness, the ideas used in designing the span program could be turned into a quantum algorithm for $3$-distinctness with time complexity matching its query complexity up to logarithmic factors \cite{BCJKM13}. 

In this work, we consider new ways of designing quantum algorithms via span programs. Consider Grover's quantum search algorithm, which, on input $x\in\{0,1\}^n$, decides if there is some $i\in [n]$ such that $x_i=1$ using only $O(\sqrt{n})$ quantum operations \cite{gro96}. The ideas behind this algorithm have been used in innumerable contexts, but in particular, a careful analysis of the ideas behind Grover's algorithm led to algorithms for similar problems, including a class of \emph{threshold functions}: given $x\in \{0,1\}^n$, decide if $|x|\geq t$ or $|x|<\eps t$, where $|x|$ denotes the Hamming weight; and approximate counting: given $x\in\{0,1\}^n$, output an estimate of $|x|$ to some desired accuracy. The results in this paper offer the possibility of obtaining analogous results for any span program. That is, given a span program for some problem $f$, our results show that one can obtain, not only an algorithm for $f$, but algorithms for a related class of threshold functions, as well as an algorithm for estimating a quantity called the \emph{span program witness size}, which is analogous to $|x|$ in the above example (and is in fact exactly $1/|x|$ in the  span program for the OR function --- see \sec{example}).

\paragraph{New Algorithms from Span Programs} We give several new means of constructing quantum algorithms from span programs. Roughly speaking, a span program can be turned into a quantum algorithm that decides between two types of inputs: those that ``hit'' a certain ``target vector'', and those that don't. We show how to turn a span program into an algorithm that decides between inputs that get ``close to'' the target vector, and those that don't. Whereas as traditionally a span program has been associated with some decision problem, this allows us to now associate, with one span program, a whole class of threshold problems. 

In addition, for any span program $P$, we can construct a quantum algorithm that estimates the \emph{positive witness size}, $w_+(x)$, to accuracy $\eps$ in $\frac{1}{\eps^{3/2}}\sqrt{w_+(x)\widetilde{W}_-}$ queries, where $\widetilde{W}_-$ is the \emph{approximate negative witness complexity} of $P$. This construction is useful whenever we can construct a span program for which $w_+(x)$ corresponds to some function we care to estimate, as is the case with the span program for OR, in which $w_+(x)=\frac{1}{|x|}$, or the span from for $st$-connectivity, in which $w_+(G)=\frac{1}{2}R_{s,t}(G)$, where $G$ is a graph, and $R_{s,t}(G)$ is the \emph{effective resistance} between $s$ and $t$ in $G$. 
We show similar results for estimating the negative witness size as well.

\paragraph{Structural Results} 
Our analysis of the structure of span programs increases the theoretical understanding of this important model. One implication of this is alternative algorithms for estimating the witness size when the phase gap (or spectral gap) of a certain unitary associated with the span program can be lower bounded. This is in contrast to previous span program algorithms, including those mentioned in the previous paragraph, which have all relied on \emph{effective spectral gap} analysis. We show how the phase gap can be lower bounded by $\frac{\sigma_{\max}(A)}{\sigma_{\min}(A(x))}$, where $A$ and $A(x)$ are linear operators associated with the span program and some input $x$, and $\sigma_{\min}$ and $\sigma_{max}$ are the smallest and largest nonzero singular values.

In addition, our exposition highlights the relationship between span programs and estimating the size of the smallest solution to a linear system, which is a problem solved by \cite{HHL09}. It is not yet clear if this relationship can lead to new algorithms, but it is an interesting direction for future work, which we discuss in \sec{open}.

\paragraph{Application to Effective Resistance} An immediate application of our results is a quantum algorithm for estimating the effective resistance between two vertices in a graph, $R_{s,t}(G)$. This example is immediate, because in \cite{BR12}, a span program for $st$-connectivity was presented, in which the positive witness size corresponds to $R_{s,t}(G)$. The results of \cite{BR12}, combined with our new span program algorithms, immediately yield an upper bound of $\tO(\frac{1}{\eps^{3/2}}n\sqrt{R_{s,t}(G)})$ for estimating the effective resistance to relative accuracy $\eps$. This upper bound also holds for time complexity, due to the time complexity analysis of \cite{BR12}.  Using our new spectral analysis techniques, we are also able to get an often better upper bound of $\tO\(\frac{1}{\eps}n\sqrt{{R_{s,t}(G)}/{\mu}}\)$, on the time complexity of estimating effective resistance, where $\mu$ is a lower bound on $\lambda_2(G)$, the second smallest eigenvalue of the Laplacian. Both algorithms use $O(\log n)$ space. We also show that a linear dependence on $n$ is necessary, so our results cannot be significantly improved. 

 These are the first quantum algorithms for this problem in the adjacency query model. Previous results have studied the problem in the edge-list model \cite{Wan13}. At the end of \sec{app}, we compare the techniques used in \cite{Wan13} to those of our algorithms. 
Classically, this quantity can be computed exactly by inverting the Laplacian, which costs $O(m)=O(n^2)$, where $m$ is the number of edges in the input graph.

\paragraph{Outline} In \sec{prelim}, we describe the algorithmic subroutines and standard linear algebra that will form the basis of our algorithms. In \sec{span-old}, we review the use of span programs in the context of quantum query algorithms, followed in \sec{span-new} by our new paradigm of \emph{approximate} span programs. At this point we will be able to formally state our results about how to use span programs to construct quantum algorithms. In \sec{span-structure}, we describe the structure of span programs, giving several results that will help us develop algorithms. The new algorithms from span programs are developed in \sec{alg}, and finally, in \sec{app}, we present our applications to estimating effective resistance. In \sec{open}, we discuss open problems.

\subsection{Preliminaries}\label{sec:prelim}

To begin, we fix notation and review some concepts from linear algebra. By $\mathcal{L}(V,W)$ we denote the set of linear operators from $V$ to $W$. For any operator $A\in\mathcal{L}(V,W)$, we denote by $\mathrm{col}A$ the columnspace, $\mathrm{row}A$ the rowspace, and $\ker A$ the kernel of $A$. 

\begin{definition}[Singular value decomposition]
Any linear operator $A\in\mathcal{L}(V,W)$ can be written as $A=\sum_{i=1}^{r}\sigma_i\ket{\psi_i}\bra{\phi_i}$ for positive real numbers $\sigma_i$, called the \emph{singular values}, an orthonormal basis for $\mathrm{row}A$, $\{\ket{\phi_i}\}_i$, called the \emph{right singular vectors},  and an orthonormal basis for $\mathrm{col}A$, $\{\ket{\psi_i}\}_i$, called the \emph{left singular vectors}. We define $\sigma_\mathrm{min}(A):=\min_i\sigma_i$ and $\sigma_\mathrm{max}(A):=\max_i\sigma_i$. 
\end{definition}

\begin{definition}[Pseudo-inverse]
For any linear operator $A$ with singular value decomposition $A=\sum_{i=1}^{r}\sigma_i\ket{\psi_i}\bra{\phi_i}$, we define the \emph{pseudo-inverse} of $A$ as $A^+:=\sum_{i=1}^{r}\frac{1}{\sigma_i}\ket{\phi_i}\bra{\psi_i}$. We note that $A^+A$ is the orthogonal projector onto $\mathrm{row}A$, and $AA^+$ is the orthogonal projector onto $\mathrm{col}A$. For any $\ket{v}\in \mathrm{row}A$, the unique smallest vector $\ket{w}$ satisfying $A\ket{w}=\ket{v}$ is $A^+\ket{v}$. 
\end{definition}

\noindent The algorithms in this paper solve either decision problems, or estimation problems. 

\begin{definition}
Let $f:X\subseteq [q]^n\rightarrow \{0,1\}$. We say that an algorithm \emph{decides $f$ with bounded error} if for any $x\in X$, with probability at least $2/3$, the algorithm outputs $f(x)$ on input~$x$.
\end{definition}

\begin{definition}
Let $f:X\subseteq[q]^n\rightarrow \mathbb{R}_{\geq 0}$. We say that an algorithm \emph{estimates $f$ to relative accuracy $\eps$} with bounded error if for any $x\in X$, with probability at least $2/3$, on input $x$ the algorithm outputs $\tilde{f}$ such that $|{f(x)-\tilde{f}}|\leq \eps f(x).$
\end{definition}

We will generally omit the description ``with bounded error'', since all of our algorithms will have bounded error.

All algorithms presented in this paper are based on the following structure. We have some initial state $\ket{\phi_0}$, and some unitary operator $U$, and we want to estimate $\norm{\Pi_0\ket{\phi_0}}$, where $\Pi_0$ is the orthogonal projector onto the 1-eigenspace of $U$. The first step in this process is a quantum algorithm that estimates, in a new register, the phase of $U$ applied to the input state. 

\begin{theorem}[Phase Estimation \cite{Kit95,CEMM98}]\label{thm:phase-est}
Let $U=\sum_{j=1}^m e^{i\theta_j}\ket{\psi_j}\bra{\psi_j}$ be a unitary, with $\theta_1,\dots,\theta_m \in (-\pi,\pi]$. For any $\Theta\in (0,\pi)$ and $\eps\in (0,1)$, there exists a quantum algorithm that makes $O\(\frac{1}{\Theta}\log \frac{1}{\eps}\)$ controlled calls to $U$ and, on input $\ket{\psi_j}$, outputs a state $\ket{\psi_j}\ket{\omega}$ such that if $\theta_j=0$, then $\ket{\omega}=\ket{0}$, and if $|\theta_j|\geq \Theta$, $|\braket{0}{\omega}|^2\leq \eps$. If $U$ acts on $s$ qubits, the algorithm uses $O(s+\log \frac{1}{\Theta})$ space.
\end{theorem}

\noindent The precision needed to isolate $\Pi_0\ket{\phi_0}$ depends on the smallest nonzero phase of $U$, the \emph{phase gap}.

\begin{definition}[Phase Gap]
Let $\{e^{i\theta_j}\}_{j\in S}$ be the eigenvalues of a unitary operator $U$, with $\{\theta_j\}_{j\in S}\subset (-\pi,\pi]$. Then the \emph{phase gap} of $U$ is $\Delta(U):=\min\{|\theta_j|:\theta_j\neq 0\}$. 
\end{definition}

\noindent In order to estimate $\norm{\Pi_0\ket{\phi_0}}^2$, given a state $\ket{0}\Pi_0\ket{\phi_0}+\ket{1}(I-\Pi_0)\ket{\phi_0}$, we use the following. 

\begin{theorem}[Amplitude Estimation \cite{BHMT02}]\label{thm:amp-est}
Let $\mathcal{A}$ be a quantum algorithm that outputs $\sqrt{p(x)}\ket{0}\ket{\Psi_x(0)}+\sqrt{1-p(x)}\ket{1}\ket{\Psi_x(1)}$ on input $x$. Then there exists a quantum algorithm that estimates $p(x)$ to precision $\eps$ using $O\(\frac{1}{\eps}\frac{1}{\sqrt{p(x)}}\)$ calls to $\mathcal{A}$. 
\end{theorem}

If we know that the amplitude is either $\leq p_0$ or $\geq p_1$ for some $p_0<p_1$, then we can use amplitude estimation to distinguish between these two cases.

\begin{corollary}[Amplitude Gap]\label{cor:amp-est-dec}
Let $\mathcal{A}$ be a quantum algorithm that outputs $\sqrt{p(x)}\ket{0}\ket{\Psi_x(0)}+\sqrt{1-p(x)}\ket{1}\ket{\Psi_x(1)}$ on input $x$. For any $0\leq p_1<p_0\leq 1$, we can distinguish between the cases $p(x)\geq p_0$ and $p(x)\leq p_1$ with bounded error using $O\(\frac{\sqrt{p_0}}{p_0-p_1}\)$ calls to $\mathcal{A}$. 
\end{corollary}
\begin{proof}
By \cite[Thm. 12]{BHMT02}, using $M$ calls to $\mathcal{A}$, we can obtain an estimate $\tilde{p}$ of $p(x)$ such that 
$$\abs{\tilde{p}-p(x)}\leq \frac{2\pi \sqrt{p(x)(1-p(x))}}{M}+\frac{\pi^2}{M^2}$$
with probability $3/4$. Let $M=4\pi\frac{\sqrt{p_0+p_1}}{p_0-p_1}$. Then note that for any $x_1$ and $x_0$ such that $p(x_1)\leq p_1$ and $p(x_0)\geq p_0$, we have, using $\sqrt{p_0+p_1}\geq (\sqrt{p_0}+\sqrt{p_1})/\sqrt{2}$, 
$$M\geq 2\sqrt{2}\pi\frac{\sqrt{p_0}+\sqrt{p_1}}{p_0-p_1}
=2\sqrt{2}\pi\frac{1}{\sqrt{p_0}-\sqrt{p_1}}
\geq 2\sqrt{2}\pi\frac{1}{\sqrt{p(x_0)}-\sqrt{p(x_1)}}
=2\sqrt{2}\pi\frac{\sqrt{p(x_0)}+\sqrt{p(x_1)}}{p(x_0)-p(x_1)}.$$ 
If $\tilde{p}_1$ is the estimate obtained on input $x_1$, then we have, with probability 3/4:
$$\tilde{p}_1\leq 
p(x_1)+\frac{2\pi \sqrt{p(x_1)(1-p(x_1))}}{M}+\frac{\pi^2}{M^2}
\leq p(x_1)+\frac{\sqrt{p(x_1)}(p(x_0)-p(x_1))}{\sqrt{2}(\sqrt{p(x_0)}+\sqrt{p(x_1)})}+\frac{(p_0-p_1)^2}{16(p_0+p_1)}.$$
On the other hand, if $\tilde{p}_0$ is an estimate of $p(x_0)$, then with probability 3/4:
$$\tilde{p}_0\geq p(x_0)-\frac{{2}\pi \sqrt{p(x_0)(1-p(x_0))}}{M}-\frac{\pi^2}{M^2}
\geq p(x_0)-\frac{\sqrt{p(x_0)}(p(x_0)-p(x_1))}{\sqrt{2}(\sqrt{p(x_0)}+\sqrt{p(x_1)})}-\frac{(p_0-p_1)^2}{16(p_0+p_1)}.$$
We complete the proof by showing that $\tilde{p}_1<\tilde{p}_0$, so we can distinguish these two events. We have:
\begin{eqnarray*}
\tilde{p}_0-\tilde{p}_1 &\geq &  p(x_0)-p(x_1)-\frac{(p(x_0)-p(x_1))}{\sqrt{2}(\sqrt{p(x_0)}+\sqrt{p(x_1)})}(\sqrt{p(x_0)}+\sqrt{p(x_1)})-\frac{(p_0-p_1)^2}{8(p_0+p_1)}\\
&\geq & \(1-\frac{1}{\sqrt 2}\)(p_0-p_1)-\frac{1}{8}(p_0-p_1)\;\; \geq \;\; \frac{1}{6}(p_0-p_1)\;\; > \;\; 0. 
\end{eqnarray*}
Thus, using $4\pi\frac{\sqrt{p_0+p_1}}{p_0-p_1}=O\(\frac{\sqrt{p_0}}{p_0-p_1}\)$ calls to $\mathcal{A}$, we can distinguish between $p(x)\leq p_1$ and $p(x)\geq p_0$ with success probability 3/4. 
\end{proof}

In order to make use of phase estimation, we will need to analyze the spectrum of a particular unitary, which, in our case, consists of a pair of reflections.
The following lemma 
first appeared in this form in \cite{LMR+11}:
\begin{lemma}[Effective Spectral Gap Lemma]\label{lem:gap}
Let $U=(2\Pi_A-I)(2\Pi_B-I)$ be the product of two reflections, and let $\Pi_\Theta$ be the orthogonal projector onto $\mathrm{span}\{\ket{u}:U\ket{u}=e^{i\theta}\ket{u},|\theta|\leq \Theta\}$. Then if $\Pi_A\ket{u}=0$, $\norm{\Pi_\Theta\Pi_B\ket{u}}\leq \frac{\Theta}{2}\norm{\ket{u}}$. 
\end{lemma}

\noindent The following theorem was first used in the context of quantum algorithms by Szegedy \cite{sze04}:
\begin{theorem}[\cite{sze04}]\label{thm:szegedy}
Let $U=(2\Pi_A-I)(2\Pi_B-I)$ be a unitary on a finite inner product space $H$ containing $A=\mathrm{span}\{\ket{\psi_1},\dots,\ket{\psi_a}\}$ and $B=\mathrm{span}\{\ket{\phi_1},\dots,\ket{\phi_b}\}$.  Let $\Pi_A=\sum_{i=1}^a\ket{\psi_i}\bra{\psi_i}$ and $\Pi_B=\sum_{i=1}^b\ket{\phi_i}\bra{\phi_i}$. Let $D=\Pi_A\Pi_B$ be the discriminant of $U$, and suppose it has singular value decomposition $\sum_{j=1}^r\cos\theta_j\ket{\alpha_j}\bra{\beta_j}$,
with  
$\theta_j\in [0,\frac{\pi}{2}]$. 
Then the spectrum of $U$ is $\{e^{\pm 2i\theta_j}\}_j$. The 1-eigenspace of $U$ is $(A\cap B)\oplus (A^\bot\cap B^\bot)$ and the $-1$-eigenspace is $(A\cap B^\bot)\oplus (A^\bot\cap B)$. 
\end{theorem}

Let $\Lambda_A=\sum_{j=1}^a\ket{\psi_j}\bra{j}$ and $\Lambda_B=\sum_{j=1}^b\ket{\phi_j}\bra{j}$. We note that in the original statement of \thm{szegedy}, the discriminant is defined $D'=\Lambda_A^\dagger\Lambda_B$. However it is easy to see that $D'$ and $D$ have the same singular values: 
if $D'=\sum_i\sigma_i\ket{v_i}\bra{u_i}$ is a singular value decomposition of $D'$, then $D=\sum_i\sigma_i\Lambda_A\ket{v_i}\bra{u_i}\Lambda_B^\dagger$ is a singular value decomposition of $D$, since $\Lambda_A$ acts as an isometry on the columns of $D'$, and $\Lambda_B$ acts as an isometry on the rows of $D'$. 

The following corollary to \thm{szegedy} will be useful in the analysis of several algorithms.
\begin{corollary}[Phase gap and discriminant]\label{cor:szegedy}
Let $D$ be the discriminant of a unitary $U=(2\Pi_A-I)(2\Pi_B-I)$. Then $\Delta(-U)\geq 2\sigma_{\min}(D)$.
\end{corollary}
\begin{proof}
By \thm{szegedy}, if $\{\sigma_0=\cos\theta_0 < \sigma_1=\cos\theta_1 < \dots \sigma_m=\cos\theta_m\}$ are the singular values of $D$, for $\theta_j\in [0,\frac{\pi}{2}]$, then $U$ has phases $\{\pm 2\theta_j\}_{j=0}^m\subset [-\pi, \pi]$, and so $-U$ has phases $\{\pm 2\theta_j\mp\pi\}_{j=0}^m=\{\pm (2\theta_j-\pi)\}_{j=0}^m\subset [-\pi,\pi]$. 
Thus 
$$\Delta(-U)=\min\{|\pi-2\theta_j|:\theta_j\neq {\pi}/{2}\}=|\pi-2\cos^{-1}\min\{\sigma_j:\sigma_j\neq 0\}|=|\pi-2\cos^{-1}\sigma_{\min}(D)|.$$
We have $\theta\geq \sin\theta = \cos({\pi}/{2}-\theta)$, so $\sigma_{\min}(D)\geq \cos(\pi/2-\sigma_{\min}(D))$. Then since $\cos$ is decreasing on the interval $[0,\pi/2]$, we have 
$\cos^{-1}(\sigma_{\min}(D))\leq {\pi}/{2}-\sigma_{\min}(D)$, and thus 
$$\hspace{125pt}\Delta(-U)\geq \abs{\pi-2\({\pi}/{2}-\sigma_{\min}(D)\)}=2\sigma_{\min}(D).\hspace{100pt}\qedhere$$
\end{proof}

\section{Approximate Span Programs}\label{sec:span}

\subsection{Span Programs and Decision Problems}\label{sec:span-old}

In this section, we review the concept of span programs, and their use in quantum algorithms. 

\begin{definition}[Span Program]\label{def:span}
A span program $P=(H,V,\tau,A)$ on $[q]^n$ consists of 
\begin{enumerate}
\item finite-dimensional inner product spaces $H=H_1\oplus \dots \oplus H_n\oplus H_\mathrm{true}\oplus H_\mathrm{false}$, and $\{H_{j,a}\subseteq H_j\}_{j\in [n],a\in [q]}$ such that $H_{j,1}+\dots+H_{j,q}=H_j$,
\item a vector space $V$,
\item a \emph{target vector} $\tau\in V$, and
\item a linear operator $A\in \mathcal{L}(H,V)$.
\end{enumerate}
To each string $x\in [q]^n$, we associate a subspace $H(x):=H_{1,x_1}\oplus \dots\oplus H_{n,x_n}\oplus H_\mathrm{true}$. 
\end{definition}

Although our notation in \defin{span} deviates from previous span program definitions, the only difference in the substance of the definition is that the spaces $H_{j,a}$ and $H_{j,b}$ for $a\neq b$ need not be orthogonal in our definition. This has the effect of removing $\log q$ factors in the equivalence between span programs and the dual adversary bound (for details see \cite[Sec.\ 7.1]{jef14}). The spaces $H_{\mathrm{true}}$ and $H_{\mathrm{false}}$ can be useful for designing a span program, but are never required, since we can always add an \th{$(n+1)$} variable, set $x_{n+1}=1$, and let $H_{n+1,0}=H_{\mathrm{false}}$ and $H_{n+1,1}=H_{\mathrm{true}}$. 

 A span program on $[q]^n$ partitions $[q]^n$ into two sets: \emph{positive} inputs, which we call $P_1$, and \emph{negative} inputs, which we call $P_0$. The importance of this partition stems from the fact that a span program may be converted into a quantum algorithm for deciding this partition in the quantum query model \cite{Rei09,Rei11}. Thus, if one can construct a span program whose partition of $[q]^n$ corresponds to a problem one wants to solve, an algorithm follows. In order to describe how a span program partitions $[q]^n$ and the query complexity of the resulting algorithm, we need the concept of positive and negative witnesses and witness size.

\begin{definition}[Positive and Negative Witness]
Fix a span program $P$ on $[q]^n$, and a string $x\in [q]^n$.
We say that $\ket{w}$ is a \emph{positive witness for $x$ in $P$} if $\ket{w}\in H(x)$, and $A\ket{w}=\tau$. We define the \emph{positive witness size of $x$} as:
$$w_+(x,P)=w_+(x)=\min\{\norm{\ket{w}}^2: \ket{w}\in H(x):A\ket{w}=\tau\},$$
if there exists a positive witness for $x$, and $w_+(x)=\infty$ else.
We say that $\omega\in\mathcal{L}(V,\mathbb{R})$ is a \emph{negative witness for $x$ in $P$} if $\omega A\Pi_{H(x)} = 0$ and $\omega\tau = 1$. We define the \emph{negative witness size of $x$} as:
$$w_-(x,P)=w_-(x)=\min\{\norm{\omega A}^2:{\omega\in \mathcal{L}(V,\mathbb{R}): \omega A\Pi_{H(x)}=0, \omega\tau=1}\},$$
if there exists a negative witness, and $w_-(x)=\infty$ otherwise.
If $w_+(x)$ is finite, we say that $x$ is \emph{positive} (wrt. $P$), and if $w_-(x)$ is finite, we say that $x$ is \emph{negative}. We let $P_1$ denote the set of positive inputs, and $P_0$ the set of negative inputs for $P$.  Note that for every $x\in [q]^n$, exactly one of $w_-(x)$ and $w_+(x)$ is finite; that is, $(P_0,P_1)$ partitions $[q]^n$. 
\end{definition}

For a decision problem $f:X\subseteq [q]^n\rightarrow \{0,1\}$, we say that $P$ \emph{decides} $f$ if $f^{-1}(0)\subseteq P_0$ and $f^{-1}(1)\subseteq P_1$. In that case, we can use $P$ to construct a quantum algorithm that decides $f$.

\begin{theorem}[\cite{Rei09}]
Fix $f:X\subseteq[q]^n\rightarrow\{0,1\}$, and let $P$ be a span program on $[q]^n$ that decides $f$. Let $W_+(f,P)=\max_{x\in f^{-1}(1)}w_+(x,P)$ and $W_-(f,P)=\max_{x\in f^{-1}(0)}w_-(x,P)$. Then there exists a quantum algorithm that decides $f$ using $O(\sqrt{W_+(f,P)W_-(f,P)})$ queries.
\end{theorem}

We call $\sqrt{W_+(f,P)W_-(f,P)}$ the \emph{complexity} of $P$. It is known that for any decision problem, there exists a span program whose complexity is equal, up to constants, to its query complexity \cite{Rei09,Rei11} (\cite[Sec.\ 7.1]{jef14} removes log factors in this statement), however, it is generally a difficult task to find such an optimal span program.

\subsection{Span Programs and Approximate Decision Problems}\label{sec:span-new}

Consider a span program $P$ and $x\in P_0$. Suppose there is some $\ket{w}\in H(x)$ such that $A\ket{w}$ comes extremely close to $\tau$. We might say that $x$ is very close to being in $P_1$. If all vectors in $H(y)$ for $y\in P_0\setminus\{x\}$ are very far from $\tau$, it might be slightly more natural to consider the partition $(P_0\setminus\{x\},P_1\cup\{x\})$ rather than $(P_0,P_1)$. 

As further motivation, we mention a construction of Reichardt \cite[Sec.\ 3 of full version]{Rei09} that takes any quantum query algorithm with one-sided error, and converts it into a span program whose complexity matches the query complexity of the algorithm. The target of the span program is the vector $\ket{1,\bar{0}}$, which corresponds to a quantum state with a 1 in the answer register and 0s elsewhere. If an algorithm has no error on 1-inputs, it can be modified so that it always ends in exactly this state, by uncomputing all but the answer register. An algorithm with two-sided error cannot be turned into a span program using this construction, because there is error in the final state. This is intuitively in opposition to the evidence that span programs characterize bounded (two-sided) error quantum query complexity. The exactness required by span programs seems to contrast the spirit of non-exact quantum algorithms. 

This motivates us to consider the \emph{positive error} of an input, or how close it comes to being positive. Since there is no meaningful notion of distance in $V$, we consider closeness in $H$.

\begin{definition}[Positive Error]
For any span program $P$ on $[q]^n$, and $x\in [q]^n$, we define the \emph{positive error of $x$ in $P$} as:
$$e_+(x)=e_+(x,P):=\min\left\{\norm{\Pi_{H(x)^\bot}\ket{w}}^2:A\ket{w}=\tau\right\}.$$
Note that $e_+(x,P)=0$ if and only if $x\in P_1$. Any $\ket{w}$ such that $\norm{\Pi_{H(x)^\bot}\ket{w}}^2=e_+(x)$ is called a \emph{min-error positive witness for $x$ in $P$}. We define 
$$\tilde{w}_+(x)=\tilde{w}_+(x,P):=\min\left\{\norm{\ket{w}}^2:A\ket{w}=\tau, \norm{\Pi_{H(x)^\bot}\ket{w}}^2=e_+(x)\right\}.$$
A min-error positive witness that also minimizes $\norm{\ket{w}}^2$ is called an \emph{optimal min-error positive witness for $x$}.
\end{definition}

Note that if $x\in P_1$, then $e_+(x)=0$. In that case, a min-error positive witness for $x$ is just a positive witness, and $\tilde w_+(x)=w_+(x)$. 

We can define a similar notion for positive inputs, to measure their closeness to being negative. 

\begin{definition}[Negative Error]
For any span program $P$ on $[q]^n$ and $x\in [q]^n$, we define the \emph{negative error of $x$ in $P$} as:
$$e_-(x)=e_-(x,P):=\min\left\{\norm{\omega A\Pi_{H(x)}}^2: \omega(\tau)=1\right\}.$$
Again, $e_-(x,P)=0$ if and only if $x\in P_0$. Any $\omega$ such that $\norm{\omega A\Pi_{H(x)}}^2=e_-(x,P)$ is called a \emph{min-error negative witness for $x$ in $P$}. We define 
$$\tilde{w}_-(x)=\tilde{w}_-(x,P):=\min\left\{\norm{\omega A}^2:\omega(\tau)=1,\norm{\omega A\Pi_{H(x)}}^2=e_-(x,P)\right\}.$$
A min-error negative witness that also minimizes $\norm{\omega A}^2$ is called an \emph{optimal min-error negative witness for $x$}.
\end{definition}

It turns out that the notion of span program error has a very nice characterization as exactly the reciprocal of the witness size:
$$\forall x\in P_0,\;w_-(x)=\frac{1}{e_+(x)},\qquad\mbox{and}\qquad \forall x\in P_1,\; w_+(x)=\frac{1}{e_-(x)},$$
which we prove shortly in \thm{equiv1} and \thm{equiv2}. This is a very nice state of affairs, for a number of reasons. It allows us two ways of thinking about approximate span programs: in terms of how small the error is, or how large the witness size is. That is, we can say that an input $x\in P_0$ is \emph{almost positive} either because its positive error is small, or equivalently, because its negative witness size is large. In general, we can think of $P$ as not only partitioning $P$ into $(P_0,P_1)$, but inducing an ordering on $[q]^n$ from most negative --- smallest negative witness, or equivalently, largest positive error --- to most positive --- smallest positive witness, or equivalently, largest negative error. For example, on the domain $\{x^{(1)},\dots,x^{(6)}\}\subset [q]^n$, $P$ might induce the following ordering:
\begin{center}
\begin{tikzpicture}
\draw[<-,thick] (-5,0) -- (-.2,0);
\draw[->,thick] (.2,0)--(5,0);
\draw[dashed] (0,-.5)--(0,.5);

\node at (-4,.25) {$x^{(1)}$};
\node at (-3,.25) {$x^{(2)}$};
\node at (-1,.25) {$x^{(3)}$};
\node at (2,.25) {$x^{(4)}$};
\node at (2.7,.25) {$x^{(5)}$};
\node at (4,.25) {$x^{(6)}$};

\draw[dashed] (-2,.5)--(-2,-.2);
\draw[dashed] (3.5,.5)--(3.5,-.2);

\node at (-2.75,-.25) {\small increasing positive error/};
\node at (-2.75,-.6) {\small decreasing negative witness size};

\node at (2.75,-.25) {\small increasing negative error/};
\node at (2.75,-.6) {\small decreasing positive witness size};

\end{tikzpicture}
\end{center}
The inputs $\{x^{(1)},x^{(2)},x^{(3)}\}$ are in $P_0$, and $w_-(x^{(1)})< w_-(x^{(2)})<w_-(x^{(3)})$ (although it is generally possible for two inputs to have the same witness size). The inputs $\{x^{(4)},x^{(5)},x^{(6)}\}$ are in $P_1$, and $w_+(x^{(4)})>w_+(x^{(5)})>w_+(x^{(6)})$. The span program exactly decides the partition $(\{x^{(1)},x^{(2)},x^{(3)}\},\{x^{(4)},x^{(5)},x^{(6)}\})$, but we say it \emph{approximates} any partition that respects the ordering. If we obtain a partition by drawing a line somewhere on the left side, for example $(\{x^{(1)},x^{(2)}\},\{x^{(3)},x^{(4)},x^{(5)},x^{(6)}\})$, we say $P$ \emph{negatively} approximates the function corresponding to that partition, whereas if we obtain a partition by drawing a line on the right side, for example $(\{x^{(1)},x^{(2)},x^{(3)},x^{(4)},x^{(5)}\},\{x^{(6)}\})$, we say $P$ \emph{positively} approximates the function. 

\begin{definition}[Functions Approximately Associated with $P$]\label{def:approx}
Let $P$ be a span program on $[q]^n$, and $f:X\subseteq [q]^n\rightarrow \{0,1\}$ a decision problem. For any $\lambda\in (0,1)$, we say that $P$  \emph{positively $\lambda$-approximates $f$} if $f^{-1}(1)\subseteq P_1$, and for all $x\in f^{-1}(0)$, either $x\in P_0$, or $w_+(x,P)\geq \frac{1}{\lambda}W_+(f,P)$. 
We say that $P$ \emph{negatively $\lambda$-approximates $f$} if $f^{-1}(0)\subseteq P_0$, and for all $x\in f^{-1}(1)$, either $x\in P_1$, or $w_-(x,P)\geq \frac{1}{\lambda}W_-(f,P)$. 
If $P$ decides $f$ exactly, then both conditions hold for any value of $\lambda$, and so we can say that $P$ $0$-approximates $f$. 
\end{definition}

This allows us to consider a much broader class of functions associated with a particular span program. This association is useful, because as with the standard notion of association between a function $f$ and a span program, if a function is approximated by a span program, we can convert the span program into a quantum algorithm that decides $f$ using a number of queries related to the witness sizes. Specifically, we get the following theorem, proven in \sec{alg}.

\begin{theorem}[Approximate Span Program Decision Algorithms]\label{thm:approx-alg}
Fix $f:X\subseteq[q]^n\rightarrow\{0,1\}$, and let $P$ be a span program that positively $\lambda$-approximates $f$. Define $$W_+=W_+(f,P):=\max_{x\in f^{-1}(1)}w_+(x,P)\qquad\mbox{ and }\qquad\widetilde{W}_-=\widetilde{W}_-(f,P):=\max_{x\in f^{-1}(0)}\tilde{w}_-(x,P).$$ 
There is a quantum algorithm that decides $f$ with bounded error in $O\(\frac{\sqrt{W_+\widetilde{W}_-}}{(1-\lambda)^{3/2}}\log\frac{1}{1-\lambda}\)$ queries. 
Similarly, let $P$ be a span program that negatively $\lambda$-approximates $f$. Define 
$$W_-=W_-(f,P):=\max_{x\in f^{-1}(0)}w_-(x,P)\qquad\mbox{ and }\qquad\widetilde{W}_+=\widetilde{W}_+(f,P):=\max_{x\in f^{-1}(1)}\tilde{w}_+(x,P).$$
 There is a quantum algorithm that decides $f$ with bounded error in $O\(\frac{\sqrt{W_-\widetilde{W}_+}}{(1-\lambda)^{3/2}}\log\frac{1}{1-\lambda}\)$ queries. 
\end{theorem}

With the ability to distinguish between different witness sizes, we can obtain algorithms for estimating the witness size. 

\begin{theorem}[Witness Size Estimation Algorithm]\label{thm:est-alg}
Fix $f:X\subseteq[q]^n\rightarrow \mathbb{R}_{\geq 0}$. Let $P$ be a span program such that for all $x\in X$, $f(x)=w_+(x,P)$ and define $\widetilde{W}_-=\widetilde{W}_-(f,P)=\max_{x\in X}\tilde{w}_-(x,P)$. There exists a quantum algorithm that estimates $f$ to accuracy $\eps$ in $\tO\(\frac{1}{\eps^{3/2}}\sqrt{w_+(x)\widetilde{W}_-}\)$ queries.
Similarly, let $P$ be a span program such that for all $x\in X$, $f(x)=w_-(x,P)$ and define $\widetilde{W}_+=\widetilde{W}_+(f,P)=\max_{x\in X}\tilde{w}_+(x,P)$. Then there exists a quantum algorithm that estimates $f$ to accuracy $\eps$ in $\tO\(\frac{1}{\eps^{3/2}}\sqrt{w_-(x)\widetilde{W}_+}\)$ queries.
\end{theorem}

The algorithms of \thmthm{approx-alg}{est-alg} involve phase estimation of a particular unitary $U$, as with previous span program algorithms, in order to distinguish the 1-eigenspace of $U$ from its other eigenspaces. In general, it may not be feasible to calculate the phase gap of $U$, so for the algorithms of \thmthm{approx-alg}{est-alg}, as with previous algorithms, we use the effective spectral gap lemma to bound the overlap of a particular initial state with eigenspaces of $U$ corresponding to small phases. However, by relating the phase gap of $U$ to the spectrum of $A$ and $A(x):=A\Pi_{H(x)}$, we show how to lower bound the phase gap in some cases, which may give better results. In particular, in our application to effective resistance, it is not difficult to bound the phase gap in this way, which leads to an improved upper bound. In general we have the following theorem.

\begin{theorem}[Witness Size Estimation Algorithm Using Real Phase Gap]\label{thm:est-alg-gap}
Fix $f:X\subseteq [q]^n\rightarrow\mathbb{R}_{\geq 0}$ and let $P=(H,V,\tau,A)$ be a normalized span program (see \defin{normalized}) on $[q]^n$ such that for all $x\in X$, $f(x)=w_+(x,P)$ (resp. $f(x)=w_-(x)$).  
If $\kappa\geq \frac{\sigma_\mathrm{max}(A)}{\sigma_\mathrm{min}(A\Pi_{H(x)})}$ for all $x\in X$, then the quantum query complexity of estimating $f(x)$ to relative accuracy $\eps$ is at most $\tO\(\frac{\sqrt{f(x)}\kappa}{\eps}\)$. 
\end{theorem}

\noindent\thm{approx-alg} is proven in \sec{alg-approx}, and  \thm{est-alg} is proven in \sec{alg-est}, and \thm{est-alg-gap} is proven in \sec{alg-gap}.

\subsection{Example}\label{sec:example}

To illustrate how these ideas might be useful, we will give a brief example of how a span program that leads to an algorithm for the OR function can be combined with our results to additionally give algorithms for threshold functions and approximate counting. We define a span program $P$ on $\{0,1\}^n$ as follows:
$$V=\mathbb{R},\qquad \tau=1,\qquad H_i=H_{i,1}=\mathrm{span}\{\ket{i}\},\qquad H_{i,0}=\{0\},\qquad A=\sum_{i=1}^n\bra{i}.$$
So we have $H=\mathrm{span}\{\ket{i}:i\in [n]\}$ and $H(x)=\mathrm{span}\{\ket{i}:x_i=1\}$. It's not difficult to see that $P$ decides OR. In particular, we can see that the optimal positive witness for any $x$ such that $|x|>0$ is $\ket{w_x}=\sum_{i:x_i=1}\frac{1}{|x|}\ket{i}$. The only linear function $\omega:\mathbb{R}\rightarrow\mathbb{R}$ that maps $\tau$ to $1$ is the identity, and indeed, this is a negative witness for the string $\bar{0}=0\dots 0$, since $H(\bar{0})=\{0\}$, and so $\omega A\Pi_{H(\bar{0})}=0$. 

Let $\lambda\in (0,1)$, $t\in [n]$, and let $f$ be a threshold function defined by $f(x)=1$ if $|x|\geq t$ and $f(x)=0$ if $|x|\leq \lambda t$, with the promise that one of these conditions holds. Note that if $f(x)=1$, then $w_+(x)=\norm{\ket{w_x}}^2=\frac{1}{|x|}\leq \frac{1}{t}$, so $W_+(f,P)=\frac{1}{t}$. On the other hand, if $f(x)=0$, then $w_+(x)=\frac{1}{|x|}\geq \frac{1}{\lambda t}=\frac{1}{\lambda}W_+(f,P)$, so $P$ positively $\lambda$-approximates $f$. The only approximate negative witness is $\omega$ the identity, so we have $\widetilde{W}_-=\norm{\omega A}^2=\norm{A}^2=n$. 
By \thm{approx-alg}, there is a quantum algorithm for $f$ with query complexity $\frac{1}{(1-\lambda)^{3/2}}\sqrt{W_+\widetilde{W}_-}=\frac{1}{(1-\lambda)^{3/2}}\sqrt{n/t}$.

Furthermore, since $w_+(x)=\frac{1}{|x|}$, by \thm{est-alg}, we can estimate $\frac{1}{|x|}$ to relative accuracy $\eps$, and therefore we can estimate $|x|$ to relative accuracy $2\eps$, in quantum query complexity $\frac{1}{\eps^{3/2}}\sqrt{n/|x|}$. 

These upper bounds do not have optimal scaling in $\eps$, as the actual quantum query complexities of these problems are $\frac{1}{1-\lambda}\sqrt{n/t}$ and $\frac{1}{\eps}\sqrt{n/|x|}$ \cite{BBBV97,BHMT02,BBCMdW01}, however, using \thm{est-alg-gap}, the optimal query complexities can be recovered.

\subsection{Span Program Structure and Scaling}\label{sec:span-structure}

In this section, we present some observations about the structure of span programs that will be useful in the design and analysis of our algorithms, and for general intuition. We begin by formally stating and proving \thm{equiv1} and \thm{equiv2}, relating error to witness size.

\begin{theorem}\label{thm:equiv1}
Let $P$ be a span program on $[q]^n$ and $x\in P_0$. If $\ket{\tilde w}$ is an optimal min-error positive witness for $x$, and $\omega$ is an optimal exact negative witness for $x$, then 
$$(\omega A)^\dagger = \frac{\Pi_{H(x)^\bot}\ket{\tilde w}}{\norm{\Pi_{H(x)^\bot}\ket{\tilde w}}^2},\qquad\qquad
\mbox{and so}\qquad\qquad 
w_-(x)=\frac{1}{e_+(x)}.$$
\end{theorem}
\begin{proof}
Let $\ket{\tilde w}$ be an optimal min-error positive witness for $x$, and $\omega$ an optimal zero-error negative witness for $x$. 
We have
$(\omega A)\ket{\tilde w}=\omega\tau=1$
and furthermore, since $\omega A\Pi_{H(x)}=0$, we have $(\omega A)\Pi_{H(x)^\bot}\ket{\tilde w}=1$. Thus, write $(\omega A)^\dagger = \frac{\Pi_{H(x)^\bot}\ket{\tilde w}}{\norm{\Pi_{H(x)^\bot}\ket{\tilde w}}^2}+\ket{u}$ such that $\bra{u}\Pi_{H(x)^\bot}\ket{\tilde w}=0$. 
Define $\ket{w_{\text{err}}}=\Pi_{H(x)^\bot}\ket{\tilde w}$. We have 
$A(\ket{\tilde w}-\Pi_{\ker A}\ket{w_{\text{err}}})=A\ket{\tilde w}=\tau,$
so by assumption that $\ket{\tilde w}$ has minimal error, 
$$\norm{\Pi_{H(x)^\bot}\ket{\tilde w}}\leq \norm{\Pi_{H(x)^\bot}(\ket{\tilde w}-\Pi_{\ker A}\ket{w_{\text{err}}})} \leq \norm{\Pi_{H(x)^\bot}\ket{\tilde w}-\Pi_{\ker A}\ket{w_{\text{err}}}}=\norm{\Pi_{(\ker A)^\bot}\ket{w_{\text{err}}}},$$
so $\norm{\ket{w_{\mathrm{err}}}}\leq \norm{\Pi_{(\ker A)^\bot}\ket{w_{\mathrm{err}}}}$, and so we must have $\ket{w_{\text{err}}}\in (\ker A)^\bot$. Thus, $\ker \bra{w_{\text{err}}}\subseteq \ker A$, so by the fundamental homomorphism theorem, there exists a linear function $\bar{\omega}:\mathrm{col}A\rightarrow \mathbb{R}$ such that $\bar{\omega}A=\bra{w_\text{err}}$. Furthermore, we have 
$\bar{\omega}\tau=\bar{\omega}A\ket{\tilde w}=\bra{\tilde w}\Pi_{H(x)^\bot}\ket{\tilde w}=\norm{\Pi_{H(x)^\bot}\ket{\tilde w}}^2=e_+(x),$
so $\omega'=\frac{\bar\omega}{e_+(x)}$ has $\omega'\tau=1$. 
By the optimality of $\omega$, we must have $\norm{\omega A}^2\leq \norm{\omega' A}^2$, so
\begin{eqnarray*}
\norm{\frac{\Pi_{H(x)^\bot}\ket{\tilde w}}{e_+(x)}+\ket{u}}^2 &\leq & \norm{\frac{\Pi_{H(x)^\bot}\ket{\tilde w}}{e_+(x)}}^2
\end{eqnarray*}
and so $\ket{u}=0$. Thus $(\omega A)^\dagger = \frac{\Pi_{H(x)^\bot}\ket{\tilde w}}{e_+(x)}$ and 
$w_-(x)=\norm{\omega A}^2 = \frac{\norm{\Pi_{H(x)^\bot}\ket{\tilde w}}^2}{e_+(x)^2}=\frac{1}{e_+(x)}.$
\end{proof}

\begin{theorem}\label{thm:equiv2}
Let $P$ be a span program on $[q]^n$ and $x\in P_1$. If $\ket{w}$ is an optimal exact positive witness for $x$, and $\tilde\omega$ is an optimal min-error negative witness for $x$, then 
$$\ket{w}=\frac{\Pi_{H(x)}(\tilde\omega A)^\dagger}{\norm{\tilde\omega A\Pi_{H(x)}}^2}    \qquad\qquad\mbox{and so}\qquad\qquad     w_+(x)=\frac{1}{e_-(x)}.$$
\end{theorem}
\begin{proof}
Let $\tilde\omega$ be an optimal min-error negative witness for $x$, and define $\ket{w'}=\frac{\Pi_{H(x)}(\tilde{\omega}A)^\dagger}{\norm{\tilde{\omega}A\Pi_{H(x)}}^2}$. First note that $\ket{w'}\in H(x)$. We will show that $\ket{w'}$ is a positive witness for $x$ by showing $A\ket{w'}=\tau$. Suppose $\tau$ and $A\ket{w'}$ are linearly independent, and let $\alpha\in \mathcal{L}(V,\mathbb{R})$ be such that $\alpha(A\ket{w'})=0$ and $\alpha(\tau)=1$. Then for any $\eps\in [0,1]$, we have $(\eps\tilde\omega+(1-\eps)\alpha)\tau=1$, so by optimality of $\tilde\omega$, 
\begin{eqnarray*}
\norm{\tilde\omega A\Pi_{H(x)}}^2 &\leq & \norm{(\eps\tilde\omega+(1-\eps)\alpha) A\Pi_{H(x)}}^2\\
&=&\eps^2 \norm{{\tilde\omega A}\Pi_{H(x)}}^2+(1-\eps)^2\norm{{\alpha A}\Pi_{H(x)}}^2 \mbox{ since $\alpha(A\Pi_{H(x)}(\tilde\omega A)^\dagger)=0$}\\
(1-\eps^2)\norm{\tilde\omega A\Pi_{H(x)}}^2 &\leq & (1-\eps)^2\norm{{\alpha A}\Pi_{H(x)}}^2.
\end{eqnarray*}
This implies $\norm{\tilde\omega A\Pi_{H(x)}}\leq 0$, 
a contradiction, since $\norm{\tilde\omega A\Pi_{H(x)}}>0$. Thus, we must have $A\ket{w'}=r{\tau}$ for some scalar $r$, so $\tilde\omega(A\ket{w'})=r\tilde\omega({\tau})$. We then have 
$\tilde\omega(A\ket{w'}) = {\tilde\omega A}\frac{\Pi_{H(x)}({\tilde\omega A})^\dagger}{\norm{\tilde\omega A\Pi_{H(x)}}^2}=1,$
and so we have $r=1$, and thus $A\ket{w'}={\tau}$. 
So $\ket{w'}$ is a positive witness for $x$. Let $\ket{w}\in H(x)$ be an optimal positive witness for $x$, so  $\norm{\ket{w}}^2=w_+(x)$. We have
$$\braket{w'}{w}=\frac{\tilde\omega A\Pi_{H(x)}\ket{w}}{\norm{\tilde\omega A\Pi_{H(x)}}^2}=\frac{\tilde\omega\tau}{\norm{\tilde\omega A\Pi_{H(x)}}^2}=\frac{1}{\norm{\tilde\omega A\Pi_{H(x)}}^2}=\norm{\ket{w'}}^2.$$
Thus $\norm{\ket{w'}}^2\leq \norm{\ket{w'}}\norm{\ket{w}}$ by the Cauchy-Schwarz inequality, so since $\ket{w}$ is optimal, we must have $\norm{\ket{w}}=\norm{\ket{w'}}$. Since the the smallest $\ket{w}$ such that $A\Pi_{H(x)}\ket{w}=\tau$ is uniquely defined as $(A\Pi_{H(x)})^+\tau$, we have $\ket{w}=\ket{w'}$. Thus $w_+(x)=\norm{\ket{w}}^2=\norm{\ket{w'}}^2=\frac{1}{\norm{\tilde\omega A\Pi_{H(x)}}^2}=\frac{1}{e_-(x)}$. 
\end{proof}

\paragraph{Positive Witnesses} Fix a span program $P=(H,V,\tau,A)$ on $[q]^n$. In general, a positive witness is any $\ket{w}\in H$ such that $A\ket{w}={\tau}$. Assume the set of all such vectors is non-empty, and let $\ket{w}$ be any vector in $H$ such that $A\ket{w}={\tau}$. Then the set of positive witnesses is exactly
$$W:= \ket{w}+\ker A=\{\ket{w}+\ket{h}:\ket{h}\in \ker A\}.$$
It is well known, and a simple exercise to prove, that the unique shortest vector in $W$ is $A^+\tau$, and it is the unique vector in $W\cap (\ker A)^\bot$. We can therefore talk about the unique smallest positive witness, whenever $W$ is non-empty.

\begin{definition}\label{def:normalized}
Fix a span program $P$, and suppose $W=\{\ket{h}\in H:A\ket{h}=\tau\}$ is non-empty. We define the \emph{minimal positive witness of $P$} to be $\ket{w_0}\in W$ with smallest norm --- that is, $\ket{w_0}=A^+\tau$. We define $N_+(P):= \norm{\ket{w_0}}^2$.  
\end{definition}
Since $\ket{w_0}\in(\ker A)^\bot$, we can write any positive witness $\ket{w}$ as $\ket{w_0}+\ket{w_{0}^\bot}$ for some $\ket{w_{0}^{\bot}}\in \ker A$. If we let $T=A^{-1}(\tau)$, then we can write $T=\mathrm{span}\{\ket{w_0}\}\oplus \ker A$.

\paragraph{Negative Witnesses} Just as we can talk about a minimal positive witness, we can also talk about a minimal negative witness of $P$: any $\omega_0\in\mathcal{L}(V,\mathbb{R})$ such that $\omega_0(\tau)=1$, that minimizes $\norm{\omega_0 A}$. We define $N_-(P)=\min_{\omega_0:\omega_0(\tau)=1}\norm{\omega_0A}^2$. Note that unlike $\ket{w_0}$, $\omega_0$ might not be unique. There may be distinct $\omega_0,\omega_0'\in\mathcal{L}(V,\mathbb{R})$ that map $\tau$ to $1$ and have minimal complexity, however, one can easily show that in that case, $\omega_0A=\omega_0'A$, and that the unique globally optimal negative witness in $\mathrm{col}A$ is $\frac{\bra{\tau}}{\norm{\tau}^2}$.

For any minimal negative witness, $\omega_0$, $\omega_0A$ is conveniently related to the minimal positive witness $\ket{w_0}$ by $(\omega_0 A)^\dagger = \frac{\ket{w_0}}{N_+(P)}$, and $N_+(P)=\frac{1}{N_-(P)}$. (We leave this as an exercise, since it is straightforward to prove, and not needed for our results).

\paragraph{Span Program Scaling and Normalization} By scaling $\tau$ to get a new target $\tau'=B\tau$, we can scale a span program by an arbitrary positive real number $B$, so that all positive witnesses are scaled by $B$, and all negative witnesses are scaled by $\frac{1}{B}$. Note that this leaves $W_+W_-$ unchanged, so we can in some sense consider the span program invariant under this scaling. 

\begin{definition}
A span program $P$ is \emph{normalized} if $N_+(P)=N_-(P)=1$.
\end{definition}

Any span program can be converted to a normalized span program by replacing the target with $\tau'=\frac{\tau}{N_+}$. However, it will turn out to be desirable to normalize a span program, and also scale it, independently. We can accomplish this to some degree, as shown by the following theorem.

\begin{theorem}[Span program scaling]\label{thm:scaling}
Let $P=(H,V,\tau,A)$ be any span program on $[q]^n$, and let $N=\norm{\ket{w_0}}^2$ for $\ket{w_0}$ the minimal positive witness of $P$. For $\beta\in\mathbb{R}_{>0}$, define $P^\beta=(H^{\beta},V^{\beta},\tau^{\beta},A^{\beta})$ as follows, for $\ket{\hat{0}}$ and $\ket{\hat{1}}$ two vectors orthogonal to $H$ and $V$:
$$
\forall j\in [n], a\in [q], H_{j,a}^{\beta}:=H_{j,a},\quad
H^{\beta}_{\mathrm{true}}=H_{\mathrm{true}}\oplus \mathrm{span}\{\ket{\hat 1}\},\quad
H^{\beta}_{\mathrm{false}}=H_{\mathrm{false}}\oplus \mathrm{span}\{\ket{\hat 0}\}$$
$$V^{\beta}=V\oplus\mathrm{span}\{\ket{\hat 1}\},\quad 
A^{\beta}=\beta A+\ket{\tau}\bra{\hat 0}+\frac{\sqrt{\beta^2+N}}{\beta}\ket{\hat 1}\bra{\hat 1},\quad
\tau^{\beta}=\ket{\tau}+\ket{\hat 1}$$
Then we have the following:
\begin{itemize}
\item For all $x\in P_1$, $w_+(x,P^{\beta})=\frac{1}{\beta^2}w_+(x,P)+\frac{\beta^2}{N+\beta^2}$ and $\tilde{w}_-(x,P^{\beta})\leq \beta^2\tilde{w}_-(x,P)+2$;
\item for all $x\in P_0$, $w_-(x,P^{\beta})=\beta^2 w_-(x,P)+1$ and $\tilde{w}_+(x,P^{\beta})\leq \frac{1}{\beta^2}\tilde{w}_+(x,P)+2$;
\item the minimal witness of $P^{\beta}$ is $\ket{w_0^\beta}=\frac{\beta}{\beta^2+N}\ket{w_0}+\frac{N}{\beta^2+N}\ket{\hat 0}+\frac{\beta}{\sqrt{\beta^2+N}}\ket{\hat 1}$, and $\norm{\ket{w_0^\beta}}^2=1$.
\end{itemize}
\end{theorem}

\noindent Proof of \thm{scaling} is postponed to \app{scaling}, as it consists of straightforward computation.

\section{Span Program Algorithms}\label{sec:alg}

In this section we describe several ways in which a span program can be turned into a quantum algorithm. As in the case of algorithms previously constructed from span programs, our algorithms will consist of many applications of a unitary on $H$, applied to some initial state. Unlike previous applications, we will use $\ket{w_0}$, the minimal positive witness of $P$, as the initial state, assuming $P$ is normalized so that $\norm{\ket{w_0}}=1$. This state is independent of the input, and so can be generated with 0 queries. For \emph{negative span program algorithms}, where we want to decide a function negatively approximated by $P$, we will use a unitary $U(P,x)$, defined as follows:
$$U(P,x):=(2\Pi_{\ker A}-I)(2\Pi_{H(x)}-I)=(2\Pi_{(\ker A)^\bot}-I)(2\Pi_{H(x)^\bot}-I).$$
This is similar to the unitary used in previous span program algorithms. Note that $(2\Pi_{\ker A}-I)$ is input-independent, and so can be implemented in 0 queries. However, in order to analyze the time complexity of a span program algorithm, this reflection must be implemented (as we are able to do for our applications, following \cite{BR12}). The reflection $(2\Pi_{H(x)}-I)$ depends on the input, but it is not difficult to see that it requires two queries to implement. Since our definition of span programs varies slightly from previous definitions, we provide a proof of this fact.

\begin{lemma}
The reflection $(2\Pi_{H(x)}-I)$ can be implemented using 2 queries to $x$.
\end{lemma}
\begin{proof}
For every $i\in [n]$ and $a\in [q]$, let $R_{i,a}=(I-2\Pi_{H_{i,a}^\bot\cap H_i})$, the operator that reflects every vector in $H_i$ that is orthogonal to $H_{i,a}$. This operation is input independent, and so, can be implemented in 0 queries.
For every $i\in [n]$, let $\{\ket{\psi_{i,1}},\dots,\ket{\psi_{i,m_i}}\}$ be an orthonormal basis for $H_i$. Recall that the spaces $H_i$ are orthogonal, so we can map $\ket{\psi_{i,j}}\mapsto \ket{i}\ket{\psi_{i,j}}$. Then using one query, we can map $\ket{i}\ket{\psi_{i,j}}\mapsto \ket{i}\ket{x_i}\ket{\psi_{i,j}}$. We then perform $R_{i,x_i}$ on the last register, conditioned on the first two registers, and then uncompute the first two registers, using one additional query.
\end{proof}

For \emph{positive span program algorithms}, where we want to decide a function positively approximated by $P$, or estimate the positive witness size, we will use a slightly different unitary:
$$U'(P,x)=(2\Pi_{H(x)}-I)(2\Pi_{T}-I),$$
where $T=\ker A\oplus\mathrm{span}\{\ket{w_0}\}$, the span of positive witnesses. We have $U'=U^\dagger(I-2\ket{w_0}\bra{w_0})$. 

We begin by analyzing the overlap of the initial state, $\ket{w_0}$, with the phase spaces of the unitaries $U$ and $U'$ in \sec{spectral}. In particular, we show that the projections of $\ket{w_0}$ onto the 0-phase spaces of $U$ and $U'$ are exactly related to the witness size. Using the effective spectral gap lemma (\lem{gap}), we show that the overlap of $\ket{w_0}$ with small nonzero phase spaces is not too large.
Using this analysis, in \sec{alg-approx}, we describe how to convert a span program into an algorithm for any decision problem that is approximated by the span program, proving \thm{approx-alg}, and in \sec{alg-est}, we describe how to convert a span program into an algorithm that estimates the span program witness size, proving \thm{est-alg}. 

Finally, in \sec{alg-gap}, we give a lower bound on the phase gap of $U$ in terms of the spectra of $A$ and $A(x)=A\Pi_{H(x)}$, giving an alternative analysis to the effective spectral gap analysis of \sec{spectral} that may be better in some cases, and 
proving \thm{est-alg-gap}.

\subsection{Analysis}\label{sec:spectral}

\paragraph{Negative Span Programs} 
In this section we analyze the overlap of $\ket{w_0}$ with the eigenspaces of $U(P,x)$. For any angle $\Theta\in [0,\pi)$, we define $\Pi_\Theta^x$ as the orthogonal projector onto the $e^{i\theta}$-eigenspaces of $U(P,x)$ for which $|\theta|\leq \Theta$.

\begin{lemma}\label{lem:pi-theta-neg}
Let $P$ be a normalized span program on $[q]^n$. For any $x\in [q]^n$, 
$$\norm{\Pi_\Theta^x\ket{w_0}}^2\leq \frac{\Theta^2}{4}\tilde{w}_+(x)+\frac{1}{w_-(x)}.$$
In particular, for any $x\in P_1$, $\norm{\Pi_\Theta^x\ket{w_0}}^2\leq \frac{\Theta^2}{4}w_+(x)$. 
\end{lemma}
\begin{proof}
Suppose $x\in P_1$, and let $\ket{w_x}$ be an optimal exact positive witness for $x$, so $\Pi_{(\ker A)^\bot}\ket{w_x}=\ket{w_0}$. Then since $\Pi_{H(x)^\bot}\ket{w_x}=0$, we have by the effective spectral gap lemma (\lem{gap}):
$$\norm{\Pi_\Theta^x\ket{w_0}}^2=\norm{\Pi_\Theta^x\Pi_{(\ker A)^\bot}\ket{w_x}}^2\leq \frac{\Theta^2}{4}\norm{\ket{w_x}}^2=\frac{\Theta^2}{4}w_+(x).$$

Suppose $x\in P_0$ and let $\omega_x$ be an optimal zero-error negative witness for $x$ and $\ket{\tilde w_x}$ an optimal min-error positive witness for $x$.
First note that $\Pi_{(\ker A)^\bot}\ket{\tilde w_x}=\ket{w_0}$, so $\Pi_{(\ker A)^\bot}\Pi_{H(x)}\ket{\tilde w_x}+\Pi_{(\ker A)^\bot}\Pi_{H(x)^\bot}\ket{\tilde w_x}=\ket{w_0}$. Since $\Pi_{H(x)^\bot}\(\Pi_{H(x)}\ket{\tilde w_x}\)=0$, we have, by \lem{gap},
\begin{eqnarray*}
\norm{\Pi_\Theta\Pi_{(\ker A)^\bot}\Pi_{H(x)}\ket{\tilde w_x}}^2 &\leq & \frac{\Theta^2}{4}\norm{\Pi_{H(x)}\ket{\tilde w_x}}^2\\
\norm{\Pi_\Theta\(\ket{w_0}-\Pi_{(\ker A)^\bot}\Pi_{H(x)^\bot}\ket{\tilde w_x}\)}^2 &\leq & \frac{\Theta^2}{4}\norm{\ket{\tilde w_x}}^2\\
\norm{\Pi_\Theta\(\ket{w_0}-\Pi_{(\ker A)^\bot}\frac{(\omega_x A)^\dagger}{w_-(x)}\)}^2 &\leq & \frac{\Theta^2}{4}\norm{\ket{\tilde w_x}}^2.
\end{eqnarray*}
In the last step, we used the fact that $\frac{(\omega_x A)^\dagger}{w_-(x)}=\Pi_{H(x)^\bot}\ket{\tilde w_x}$, by \thm{equiv1}. Next note that $\Pi_{(\ker A)^\bot} (\omega_x A)^\dagger = (\omega_x A)^\dagger$ and $\Pi_{H(x)^\bot}(\omega_x A)^\dagger = (\omega_x A)^\dagger$, so $U(\omega_x A)^\dagger = (\omega_x A)^\dagger$, and therefore, $\Pi_\Theta(\omega_x A)^\dagger = (\omega_x A)^\dagger$. Thus:
\begin{eqnarray*}
\norm{\Pi_\Theta\ket{w_0}-\frac{(\omega_x A)^\dagger}{w_-(x)}}^2 &\leq & \frac{\Theta^2}{4}\norm{\ket{\tilde w_x}}^2\\
\norm{\Pi_\Theta\ket{w_0}}^2+\frac{1}{w_-(x)}-2\frac{1}{w_-(x)}\bra{w_0}\Pi_\Theta (\omega_x A)^\dagger &\leq & \frac{\Theta^2}{4}\tilde w_+(x)\\
\norm{\Pi_\Theta\ket{w_0}}^2+\frac{1}{w_-(x)}-2\frac{1}{w_-(x)} (\omega_x A\ket{w_0})^\dagger &\leq & \frac{\Theta^2}{4}\tilde w_+(x)\\
\norm{\Pi_\Theta\ket{w_0}}^2+\frac{1}{w_-(x)}-2\frac{1}{w_-(x)} (\omega_x \tau)^\dagger &\leq & \frac{\Theta^2}{4}\tilde w_+(x)\\
\norm{\Pi_\Theta\ket{w_0}}^2 &\leq & \frac{\Theta^2}{4}\tilde w_+(x) +\frac{1}{w_-(x)},
\end{eqnarray*}
where in the last line we used the fact that $\omega_x\tau=1$. 
\end{proof}

\begin{lemma}\label{lem:pi-0-neg}
Let $P$ be a normalized span program on $[q]^n$. For any $x\in [q]^n$, 
$$\norm{\Pi_0^x\ket{w_0}}^2= \frac{1}{w_-(x)}.$$
In particular, for any $x\in P_1$, $\norm{\Pi_0^x\ket{w_0}}=0$.
\end{lemma}
\begin{proof}
By \lem{pi-theta-neg}, we have $\norm{\Pi^x_0\ket{w_0}}^2\leq \frac{1}{w_-(x)}$. To see the other direction, let $\omega_x$ be an optimal zero-error negative witness for $x$ (if none exists, then $w_-(x)=\infty$ and the statement is vacuously true). Define $\ket{u}=(\omega_x A)^\dagger$. By the proof of \lem{pi-theta-neg}, 
$U\ket{u}=\ket{u}$. We have 
$\braket{u}{w_0}=\omega_x A\ket{w_0}=\omega_x\tau=1$ \mbox{and} $\norm{\ket{u}}^2=\norm{\omega_x A}^2=w_-(x),$
so we have:
$\norm{\Pi^x_0\ket{w_0}}^2\geq \norm{\frac{\ket{u}\bra{u}}{\norm{\ket{u}}^2}\ket{w_0}}^2=\frac{1}{w_-(x)}$.
\end{proof}

\paragraph{Positive Span Programs} We now prove results analogous to \lemlem{pi-theta-neg}{pi-0-neg} for the unitary $U'(P,x)$. 
For any angle $\Theta\in [0,\pi)$, we define $\overline{\Pi}^x_{\Theta}$ as the projector onto the $\theta$-phase spaces of $U'(P,x)$ for which $|\theta|\leq \Theta$. 

\begin{lemma}\label{lem:pi-theta-pos}
Let $P$ be a normalized span program on $[q]^n$. For any $x\in [q]^n$,
$$\norm{\overline{\Pi}^x_\Theta\ket{w_0}}^2\leq \frac{\Theta^2}{4}\tilde{w}_-(x)+\frac{1}{w_+(x)}.$$
In particular, if $x\in P_0$, then $\norm{\overline{\Pi}^x_\Theta\ket{w_0}}^2\leq \frac{\Theta^2}{4}w_-(x)$. 
\end{lemma}
\begin{proof}
If $x\in P_0$, then let $\omega_x$ be an optimal exact negative witness for $x$, so $\omega_x A\Pi_{H(x)}=0$, and we thus have, by the effective spectral gap lemma (\lem{gap}), 
$$\norm{\overline{\Pi}_\Theta^x\Pi_T(\omega_x A)^\dagger}^2\leq \frac{\Theta^2}{4}\norm{\omega_x A}^2=\frac{\Theta^2}{4}w_-(x).$$
We have $\omega_x A\Pi_T=\omega_x A(\Pi_{\ker A}+\ket{w_0}\bra{w_0})=\omega_x A\ket{w_0}\bra{w_0}=\omega_x\tau\bra{w_0}=\bra{w_0}$, so $\norm{\overline{\Pi}_\Theta^x\ket{w_0}}^2\leq \frac{\Theta^2}{4}w_-(x)$. 

Suppose $x\in P_1$, and let $\ket{w_x}$ be an optimal zero-error positive witness for $x$, and $\tilde\omega_x$ an optimal min-error negative witness for $x$. By \thm{equiv2}, we have $\frac{\ket{w_x}}{w_+(x)}=\Pi_{H(x)}(\tilde\omega_x A)^\dagger$. 
Since $\Pi_{H(x)}(\tilde\omega_x A\Pi_{H(x)^\bot})^\dagger=0$, we have, by \lem{gap},
\begin{eqnarray*}
\norm{\overline{\Pi}^x_\Theta\Pi_T(\tilde\omega_x A\Pi_{H(x)^\bot})^\dagger}^2 &\leq & \frac{\Theta^2}{4}\norm{\tilde\omega_x A\Pi_{H(x)^\bot}}^2\\
\norm{\overline{\Pi}^x_\Theta\Pi_T\((\tilde\omega_x A)^\dagger - \frac{\ket{w_x}}{w_+(x)}\)}^2 & \leq & \frac{\Theta^2}{4}\norm{\tilde\omega_x A}^2\\
\norm{\overline{\Pi}^x_\Theta\Pi_T(\tilde\omega_x A)^\dagger - \frac{\ket{w_x}}{w_+(x)}}^2 & \leq & \frac{\Theta^2}{4}\tilde w_-(x).\\
\end{eqnarray*}
In the last line we used the fact that $\Pi_T\ket{w_x}=\Pi_{H(x)}\ket{w_x}=\ket{w_x}$, so $U'\ket{w_x}=\ket{w_x}$, and thus \mbox{$\overline{\Pi}^x_\Theta\ket{w_x}=\ket{w_x}$}. 

Note that $\tilde\omega_xA\Pi_T=\tilde\omega_xA(\Pi_{\ker A}+\ket{w_0}\bra{w_0})=\tilde\omega_xA\ket{w_0}\bra{w_0}=\tilde\omega_x\tau\bra{w_0}=\bra{w_0}$. 
Thus, we can continue from above as:
\begin{eqnarray*}
\norm{\overline{\Pi}^x_\Theta\ket{w_0} - \frac{\ket{w_x}}{w_+(x)}}^2 & \leq & \frac{\Theta^2}{4}\tilde w_-(x)\\
\norm{\overline{\Pi}^x_\Theta\ket{w_0}}^2+\norm{\frac{\ket{w_x}}{w_+(x)}}^2 -\frac{2}{w_+(x)}\bra{w_0}\overline{\Pi}^x_\Theta\ket{w_x}& \leq & \frac{\Theta^2}{4}\tilde w_-(x)\\
\norm{\overline{\Pi}^x_\Theta\ket{w_0}}^2+\frac{1}{w_+(x)} -\frac{2}{w_+(x)}\braket{w_0}{w_x}& \leq & \frac{\Theta^2}{4}\tilde w_-(x)\\
\norm{\overline{\Pi}^x_\Theta\ket{w_0}}^2& \leq & \frac{\Theta^2}{4}\tilde w_-(x)+\frac{1}{w_+(x)},\\
\end{eqnarray*}
where in the last line we used the fact that $\braket{w_0}{w_x}=1$. 
\end{proof}

\begin{lemma}\label{lem:pi-0-pos}
Let $P$ be a normalized span program on $[q]^n$. For any $x\in [q]^n$,
$$\norm{\overline{\Pi}^x_0\ket{w_0}}^2=\frac{1}{w_+(x)}.$$
In particular, if $x\in P_0$, then $\norm{\overline{\Pi}_0^x\ket{w_0}}=0$. 
\end{lemma}
\begin{proof} 
By \lem{pi-theta-pos}, $\norm{\overline{\Pi}^x_0\ket{w_0}}^2\leq \frac{1}{w_+(x)}$. 
Let $\ket{w_x}=\ket{w_0}+\ket{w_0^\bot}$ be an optimal zero-error positive witness for $x$. Since $\ket{w_x}\in H(x)\cap T$, $U'\ket{w_x}=\ket{w_x}$, so
$\norm{\overline{\Pi}^x_0\ket{w_0}}^2 \geq \frac{\braket{w_x}{w_0}}{\norm{\ket{w_x}}^2}\; \geq \; \frac{1}{w_+(x)}$.
\end{proof}

\subsection{Algorithms for Approximate Span Programs}\label{sec:alg-approx}

Using the spectral analysis from \sec{spectral}, we can design an algorithm that decides a function that is approximated by a span program. 
We will give details for the negative case, using \lemlem{pi-theta-neg}{pi-0-neg}. A nearly identical argument proves the analogous statement for the positive case, using \lemlem{pi-theta-pos}{pi-0-pos} instead. 

Throughout this section, fix a decision problem $f$ on $[q]^n$, and let $P$ be a normalized span program that negatively $\lambda$-approximates $f$.
By \lemlem{pi-0-neg}{pi-theta-neg}, it is possible to distinguish between the cases $f(x)=0$, in which $\frac{1}{w_-(x)}\geq \frac{1}{W_-}$, and $f(x)=1$, in which $\frac{1}{w_-(x)}\leq \frac{\lambda}{W_-}$ using phase estimation to sufficient precision, and amplitude estimation on a 0 in the phase register. We give details in the following theorem.

\begin{lemma}\label{lem:alg-approx-norm}
Let $P$ be a normalized $\lambda$-negative approximate span program for $f$. Then the quantum query complexity of $f$ is at most $O\(\frac{1}{(1-\lambda)^{3/2}}W_-\sqrt{\widetilde{W}_+}\log\frac{W_-}{1-\lambda}\)$.
\end{lemma}
\begin{proof}
Let $U(P,x)=\sum_{j=1}^me^{i\theta_j}\ket{\psi_j}\bra{\psi_j}$, and let $\ket{w_0}=\sum_{j=1}^m\alpha_j\ket{\psi_j}$. Then applying phase estimation (\thm{phase-est}) to precision $\Theta=\sqrt{\frac{4(1-\lambda)}{3W_-\widetilde{W}_+}}$ and error $\eps=\frac{1}{6}\frac{1-\lambda}{W_-}$ produces a state
$\ket{w_0'}=\sum_{j=1}^m\alpha_j\ket{\psi_j}\ket{\omega_j}$
such that if $\theta_j=0$, then $\ket{\omega_j}=\ket{0}$, and if $|\theta_j|>\Theta$ then $|\braket{\omega_j}{0}|^2\leq \eps$.
Let $\Lambda_0$ be the projector onto states with $0$ in the phase register. We have: 
$\norm{\Lambda_0\ket{w_0'}}^2=\sum_{j=1}^m|\alpha_j|^2|\braket{0}{\omega_j}|^2.$ 
Suppose $x\in f^{-1}(0)$, so $\norm{\Pi_0^x\ket{w_0}}^2=\sum_{j:\theta_j=0}|\alpha_j|^2\geq \frac{1}{w_-(x)}$, by \lem{pi-0-neg}, and thus we have:
$$\norm{\Lambda_0\ket{w_0'}}^2\geq \sum_{j:\theta_j=0}|\alpha_j|^2|\braket{0}{0}|^2=\norm{\Pi_0^x\ket{w_0}}^2\geq\frac{1}{w_-(x)}\geq \frac{1}{W_-}=:p_0.$$
On the other hand, suppose $x\in f^{-1}(1)$. Since $P$ negatively $\lambda$-approximates $f$ and $x\in f^{-1}(1)$, $w_-(x,P)\geq \frac{1}{\lambda}W_+(x,P)$. By \lem{pi-theta-neg}, we have 
$$\norm{\Pi_\Theta^x\ket{w_0}}^2 \leq \frac{1}{w_-(x,P)}+\frac{\Theta^2}{4}\tilde{w}_+(x,P)
\leq \frac{\lambda}{W_-}+\frac{1-\lambda}{3W_-\widetilde{W}_+}\widetilde{W}_+=\frac{1}{3}\frac{1+2\lambda}{W_-}$$
and thus
$$\norm{\Lambda_0\ket{w_0'}}^2\leq \!\!\sum_{j:|\theta_j|\leq \Theta}\!\!|\alpha_j|^2+\!\!\sum_{j:|\theta_j|>\Theta}\!\!|\alpha_j|^2|\braket{\omega_j}{0}|^2=\norm{\Pi_\Theta^x\ket{w_0}}^2+\eps\!\!\sum_{j:|\theta_j|>\Theta}\!\!|\alpha_j|^2\leq \frac{1+2\lambda}{3W_-}+\frac{1-\lambda}{6W_-}=:p_1.$$
By \cor{amp-est-dec}, we can distinguish between these cases using $O\(\frac{\sqrt{p_0}}{p_0-p_1}\)$ calls to phase estimation, which costs $\frac{1}{\Theta}\log\frac{1}{\eps}$. In this case, we have
$$p_0-p_1=\frac{1-\frac{1}{3}-\frac{2}{3}\lambda-\frac{1}{6}+\frac{1}{6}\lambda}{W_-}=\frac{1}{2}\frac{1-\lambda}{W_-}.$$
The total number of calls to $U$ is:
$$\hspace{54pt}\frac{\sqrt{p_0}}{p_0-p_1}\frac{1}{\Theta}\log\frac{1}{\eps}=\frac{{W_-}}{\sqrt{W_-}(1-\lambda)}\sqrt{\frac{W_-\widetilde{W}_+}{1-\lambda}}\log\frac{W_-}{1-\lambda}=\frac{W_-\sqrt{\widetilde{W}_+}}{(1-\lambda)^{3/2}}\log \frac{W_-}{1-\lambda}.\hspace{39pt}\qedhere$$
\end{proof}

In addition to wanting to extend this to non-normalized span programs, we note that this expression is not symmetric in the positive and negative error. Using \thm{scaling}, we can normalize  any span program, while also scaling the positive and negative witnesses. 
This gives us the following.

\begin{corollary}\label{cor:alg-approx-neg}
Let $P$ be any span program that negatively $\lambda$-approximates $f$. Then the quantum query complexity of $f$ is at most $O\(\frac{1}{(1-\lambda)^{3/2}}\sqrt{W_-(f,P)\widetilde{W}_+(f,P)}\log\frac{1}{1-\lambda}\)$.
\end{corollary}
\begin{proof}
We will use the scaled span program described in \thm{scaling}. 
Let $\beta=\frac{1}{\sqrt{W_-(f,P)}}$. Then $P^\beta$ is a normalized span program with 
$$W_-(f,P^\beta)=\max_{x\in f^{-1}(0)}w_-(x,P^\beta)= {\beta^2} \max_{x\in f^{-1}(0)}w_-(x,P)+1= \frac{1}{W_-}W_-+1=2,$$
$$\mbox{and}\quad\widetilde{W}_+(f,P^\beta)=\max_{x\in f^{-1}(1)}\tilde{w}_+(x,P^\beta)\leq \frac{1}{\beta^2} \max_{x\in f^{-1}(1)}\tilde{w}_+(x,P)+2= W_-(f,P)\widetilde{W}_+(f,P)+2.$$
If we define $\lambda^{(\beta)}:=\frac{\max_{x\in f^{-1}(0)}w_-(x,P^\beta)}{\min_{x\in f^{-1}(1)}w_-(x,P^\beta)}=\frac{\beta^2W_-(f,P)+1}{\beta^2\frac{1}{\lambda}W_-(f,P)+1}=\frac{2}{\frac{1}{\lambda}+1}$, then clearly $P^{\beta}$ negatively $\lambda^{(\beta)}$-approximates $f$, so we can apply \lem{alg-approx-norm}. We have 
$\frac{1}{1-\lambda^{(\beta)}}=\frac{1}{1-\frac{2\lambda}{1+\lambda}}=\frac{1+\lambda}{1-\lambda}$
 so we can decide $f$ in query complexity (neglecting constants):
$$\!\(\frac{1+\lambda}{1-\lambda}\)^{\frac{3}{2}}\!\!\sqrt{2\(W_-(f,P)\widetilde{W}_+(f,P)+2\)}\log2\frac{1+\lambda}{1-\lambda}=
\frac{1}{(1-\lambda)^{\frac{3}{2}}}\sqrt{W_-(f,P)\widetilde{W}_+(f,P)}\log\frac{1}{1-\lambda}.\qedhere$$
\end{proof}

By computations analogous to \lem{alg-approx-norm} and \cor{alg-approx-neg} (using $\beta=\sqrt{W_+}$), we can show that if $P$ positively $\lambda$-approximates $f$, then $f$ has quantum query complexity 
 $O\(\frac{1}{(1-\lambda)^{3/2}}\sqrt{W_+\widetilde{W}_-}\log\frac{1}{1-\lambda}\)$. This and \cor{alg-approx-neg} imply \thm{approx-alg}.

\subsection{Estimating the Witness Size}\label{sec:alg-est}

Using the algorithms for deciding approximate span programs (\thm{approx-alg}) as a black box, we can construct a quantum algorithm that estimates the positive or negative witness size of an input using standard algorithmic techniques. We give the full proof for the case of positive witness size, as negative witness size is virtually identical. This proves \thm{est-alg}.

\begin{theorem}[Estimating the Witness Size]\label{thm:witness-est}
Fix $f:X\subseteq [q]^n\rightarrow \mathbb{R}_{>0}$. Let $P$ be a span program on $[q]^n$ such that for all $x\in X$, $f(x)=w_+(x,P)$. Then the quantum query complexity of estimating $f$ to accuracy $\eps$ is $\tO\(\frac{\sqrt{w_+(x)\widetilde{W}_-(P)}}{\eps^{3/2}}\)$.  
\end{theorem}
\begin{proof}
We will estimate $e(x)=\frac{1}{w_+(x)}$. The basic idea is to use the algorithm from \thm{approx-alg} to narrow down the interval in which the value of $e(x)$ may lie. Assuming that the span program is normalized (which is without loss of generality, since normalizing by scaling $\tau$ does not impact relative accuracy) we can begin with the interval $[0,1]$. We stop when we reach an interval $[e_{\min},e_{\max}]$ such that the midpoint $\tilde{e}=\frac{e_{\max}+e_{\min}}{2}$ satisfies $(1-\eps)e_{\max}\leq \tilde{e}\leq (1+\eps)e_{\min}$.

Let $\mathtt{Decide}(P,w,\lambda)$ be the quantum algorithm from \thm{approx-alg} that decides the (partial) function $g:P_1\rightarrow\{0,1\}$ defined by $g(x)=1$ if $w_+(x)\leq w$ and $g(x)=0$ if $w_+(x)\geq \frac{w}{\lambda}$. We will amplify the success probability so that with high probability, $\mathtt{Decide}$ returns $g(x)$ correctly every time it is called by the algorithm, and we will assume that this is the case. The full witness estimation algorithm consists of repeated calls to $\mathtt{Decide}$ as follows:
\vspace{2pt}
\hrule
\vspace{2pt}
\noindent$\mathtt{WitnessEstimate}(P,\eps)$:
\vspace{2pt}
\hrule
\begin{enumerate}
\item $e_{\max}^{(1)}=1$, $e_{\min}^{(1)}=0$, $e_1^{(1)}=\frac{2}{3}$, $e_0^{(1)}=\frac{1}{3}$
\item For $i=1,2,\dots$ repeat:
	\begin{enumerate}
	\item Run $\mathtt{Decide}(P,w,\lambda)$ with $w={1}/{e_1^{(i)}}$ and $\lambda={e_0^{(i)}}/{e_1^{(i)}}$. 
	\item If $\mathtt{Decide}$ outputs $1$, indicating $w_+(x)\leq w$, set $e_{\max}^{(i+1)}=e_{\max}^{(i)}$ and $e_{\min}^{(i+1)}=e_0^{(i)}$.
	\item Else, set $e_{\min}^{(i+1)}=e_{\min}^{(i)}$ and $e_{\max}^{(i+1)}=e_1^{(i)}$.
	\item If $e_{\max}^{(i+1)}\leq (1+\eps)e_{\min}^{(i+1)}$, return $\tilde{e}=\frac{e_{\max}^{(i+1)}+e^{(i+1)}_{\min}}{2}$. 
	\item Else, set $e_1^{(i+1)}=\frac{2}{3}e_{\max}^{(i+1)}+\frac{1}{3}e_{\min}^{(i+1)}$ and $e_0^{(i+1)}=\frac{1}{3}e_{\max}^{(i+1)}+\frac{2}{3}e_{\min}^{(i+1)}$.
	\end{enumerate}
\end{enumerate}
\hrule
\vspace{2pt}
We can see by induction that for every $i$, $e_{\min}^{(i)}\leq\frac{1}{w_+(x)}\leq e_{\max}^{(i)}$. This is certainly true for $i=1$, since $w_+(x)\geq \norm{\ket{w_0}}^2=1$. Suppose it's true at step $i$. At step $i$ we run $\mathtt{Decide}(P,w_i,\lambda_i)$ with $w_i={1}/{e_1^{(i)}}$ and $\frac{w_i}{\lambda_i}={1}/{e_0^{(i)}}$. If $\frac{1}{w_+(x)}\geq e_1^{(1)}$, then $\mathtt{Decide}$ returns 1, so we have $\frac{1}{w_+(x)}\in [e_{0}^{(i)},e_{\max}^{(i)}]=[e_{\min}^{(i+1)},e_{\max}^{(i+1)}]$. 
If $\frac{1}{w_+(x)}\leq e_0^{(i)}$, then $\mathtt{Decide}$ returns $0$, so we have $\frac{1}{w_+(x)}\in [e_{\min}^{(i)},e_{1}^{(i)}]=[e_{\min}^{(i+1)},e_{\max}^{(i+1)}]$. Otherwise, $\frac{1}{w_+(x)}\in [e_0^{(i)},e_1^{(i)}]$, which is a subset of both $[e_{0}^{(i)},e_{\max}^{(i)}]$ and $[e_{\min}^{(i)},e_{1}^{(i)}]$, so in any case, $\frac{1}{w_+(x)}\in [e_{\min}^{(i+1)},e_{\max}^{(i+1)}]$. 

To see that the algorithm terminates, let $\Delta_i=e_{\max}^{(i)}-e_{\min}^{(i)}$ denote the length of the remaining interval at round $i$. We either have $\Delta_{i+1}=e_{\max}^{(i)}-e_0^{(i)}=e_{\max}^{(i)}-\frac{1}{3}e_{\max}^{(i)}-\frac{2}{3}e_{\min}^{(i)}=\frac{2}{3}\Delta_i$, or $\Delta_{i+1}=e_1^{(i)}-e_{\min}^{(i)}=\frac{2}{3}e_{\max}^{(i)}+\frac{1}{3}e_{\min}^{(i)}-e_{\min}^{(i)}=\frac{2}{3}\Delta_i$, so $\Delta_i=(2/3)^{i-1}$. We terminate at the smallest $T$ such that $(2/3)^{T-1}=\Delta_T=e_{\max}^{(T)}-e_{\min}^{(T)}\leq (1+\eps-1) e_{\min}^{(T)}\leq \frac{\eps}{w_+(x)}$. Thus we terminate before $T=\ceil{\log_{3/2}\frac{w_+(x)}{\eps}+1}$. 

Next, we show that, assuming $\mathtt{Decide}$ does not err, the estimate is correct to within $\eps$. 
Let $\tilde{e}=\frac{1}{2}(e^{(T)}_{\max}+e^{(T)}_{\min})$ be the returned estimate. Recall that we only terminate when $e_{\max}^{(T)}\leq (1+\eps)e_{\min}^{(T)}$. We have 
$$\frac{1}{\tilde{e}}=\frac{2}{e_{\max}^{(T)}+e_{\min}^{(T)}}\leq \frac{2}{e_{\max}^{(T)}\(1+\frac{1}{1+\eps}\)}\leq \frac{2}{\frac{1}{w_+(x)}\(\frac{2+\eps}{1+\eps}\)}\leq \(1+\eps\)w_+(x),$$
$$\mbox{and}\qquad \frac{1}{\tilde{e}}\geq \frac{2}{e_{\min}(1+1+\eps)}\geq \frac{1}{\frac{1}{w_+(x)}(1+\eps/2)}=\(1-\frac{\eps/2}{1+\eps/2}\)w_+(x)\geq \(1-\frac{\eps}{2}\)w_+(x).$$
Thus, $|1/\tilde{e}-w_+(x)|\leq \eps w_+(x)$.

By \thm{approx-alg}, $\mathtt{Decide}(P,w,\lambda)$ runs in cost $O\(\frac{\sqrt{w \widetilde{W}_-}}{(1-\lambda)^{3/2}}\log\frac{1}{1-\lambda}\)$. Let $w_i=1/e_1^{(i)}$ and $\lambda_i=e_0^{(i)}/e_1^{(i)}$ be the values used at the \th{$i$} iteration. Since $e_1^{(i)}\leq e_{\max}^{(i)}\leq \frac{1}{w_+(x)}+\Delta_i$, we have 
$$\frac{1}{1-\lambda_i}=\frac{e_1^{(i)}}{e_1^{(i)}-e_0^{(i)}}\leq \frac{\frac{1}{w_+(x)}+\Delta_i}{\frac{2}{3}e_{\max}^{(i)}+\frac{1}{3}e_{\min}^{(i)}-\frac{1}{3}e_{\max}^{(i)}-\frac{2}{3}e_{\min}^{(i)}}= \frac{3}{w_+(x)\Delta_i}+3=O(1/\eps),$$
since $\Delta_i=(2/3)^{i-1}\geq (2/3)^{T-1}=\Omega\(\frac{\eps}{w_+(x)}\)$. 
Observe $\frac{\sqrt{w_i}}{(1-\lambda_i)^{3/2}}=\frac{e_1^{(i)}}{(e_1^{(i)}-e_0^{(i)})^{3/2}}\leq \(\frac{1}{w_+(x)}+\Delta_i\)\frac{3}{\Delta_i^{3/2}},$
so, ignoring the $\log\frac{1}{1-\lambda_i}=O(\log\frac{1}{\eps})$ factor, the cost of the \th{$i$} iteration can be computed as:
$$C_i=\frac{\sqrt{w_i\widetilde{W}_-}}{(1-\lambda_i)^{3/2}}
\leq  \sqrt{\widetilde{W}_-}\(\frac{1}{w_+(x)}+\Delta_i\)\frac{3}{\Delta_i^{3/2}}=3\frac{\sqrt{\widetilde{W}_-}}{w_+(x)}\(\frac{3}{2}\)^{\frac{3}{2}(i-1)}+3\sqrt{\widetilde{W}_-}\(\frac{3}{2}\)^{\frac{1}{2}(i-1)}.$$
We can thus compute the total cost (neglecting logarithmic factors):
\begin{eqnarray*}
\sum_{i=1}^TC_i\! &\leq & \!\frac{\sqrt{\widetilde{W}_-}}{w_+(x)}\sum_{i=1}^T\(\frac{3}{2}\)^{\frac{3}{2}(i-1)}+\sqrt{\widetilde{W}_-}\sum_{i=1}^T\({\frac{3}{2}}\)^{\frac{1}{2}(i-1)}
\;\leq\; \frac{\sqrt{\widetilde{W}_-}}{w_+(x)}\frac{\(\frac{3}{2}\)^{\frac{3}{2}T}-1}{\(\frac{3}{2}\)^{3/2}-1}+\sqrt{\widetilde{W}_-}\frac{\(\frac{3}{2}\)^{\frac{1}{2}T}-1}{\(\frac{3}{2}\)^{1/2}-1}\\
&\leq & \!O\(\frac{\sqrt{\widetilde{W}_-}}{w_+(x)}\(\frac{w_+(x)}{\eps}\)^{3/2}+\sqrt{\widetilde{W}_-}\(\frac{w_+(x)}{\eps}\)^{1/2}\)\; = \; O\(\frac{\sqrt{\widetilde{W}_-w_+(x)}}{\eps^{3/2}}\),
\end{eqnarray*}
using the fact that $(2/3)^{T}=\Theta\(\frac{\eps}{w_+(x)}\)$.  

Finally, we have been assuming that $\mathtt{Decide}$ returns the correct bit on every call. We now justify this assumption. At round $i$, we will amplify the success probability of $\mathtt{Decide}$ to $1-\frac{1}{9}(2/3)^{i-1}$, incurring a factor of $\log(9(3/2)^{i-1})=O(\log\frac{w_+(x)}{\eps})$ in the complexity. Then the total error is at most:
$$\sum_{i=1}^T \frac{1}{9}(2/3)^{i-1}=\frac{1}{9}\frac{1-(2/3)^{T-1}}{1-\frac{2}{3}}=\frac{1}{3}\(1-\frac{\eps}{w_+(x)}\)\leq \frac{1}{3}.$$
Thus, with probability at least $2/3$, $\mathtt{Decide}$ never errs, and the algorithm is correct. 
\end{proof}

\subsection{Span Program Phase Gap}\label{sec:alg-gap}

The scaling in the error from \thm{witness-est}, ${1}/{\eps^{3/2}}$, is not ideal. For instance, we showed in \sec{example} how to construct a quantum algorithm for approximate counting based on a simple span program for the OR function with complexity that scales like ${1}/{\eps^{3/2}}$ in the error, whereas the best quantum algorithm for this task has complexity scaling as ${1}/{\eps}$ in the error. 
However, the following theorem, which is a corollary to \lem{pi-0-neg} and \lem{pi-0-pos}, gives an alternative analysis of the complexity of the algorithm in \thm{witness-est} that may be better in some cases, and in particular, has the more natural error dependence ${1}/{\eps}$.

\begin{theorem}\label{thm:gap-improvement}
Fix $f:X\subseteq [q]^n\rightarrow \mathbb{R}_{>0}$. Let $P$ be a normalized span program on $[q]^n$ such that $X\subseteq P_0$, and for all $x\in X$, $w_-(x,P)=f(x)$; and define $\Delta(f)=\min_{x\in X}\Delta(U(P,x))$. Then there is a quantum algorithm that estimates $f$ to relative accuracy $\eps$ using $\tO\(\frac{1}{\eps}\frac{\sqrt{w_-(x,P)}}{\Delta(f)}\)$ queries. Similarly, let $P$ be a normalized span program such that $X\subseteq P_1$, and for all  $x\in X$, $w_+(x,P)=f(x)$; and define $\Delta'(f)=\min_{x\in X}\Delta(U'(P,x))$. Then there is a quantum algorithm that estimates $f$ with relative accuracy $\eps$ using $\tO\(\frac{1}{\eps}\frac{\sqrt{w_+(x,P)}}{\Delta'(f)}\)$ queries. 
\end{theorem}
\begin{proof}
To estimate $w_-(x)$, we can use phase estimation of $U(P,x)$ applied to $\ket{w_0}$, with precision $\Delta=\Delta(f)$ and accuracy $\epsilon=\frac{\eps}{8}\frac{1}{W_-(P,f)}$, however, this results in $\log W_-$ factors, and $W_-$ may be significantly larger than $w_-(x)$. Instead, we will start with $\epsilon=\frac{1}{2}$, and decrease it by $1/2$ until $\epsilon\approx \frac{\eps}{w_-(x,P)}$. 

Let $\ket{w_0'}$ be the result of applying phase estimation to precision $\Delta=\Delta(f)$ and accuracy $\epsilon$, and let $\Lambda_0$ be the projector onto states with 0 in the phase register. We will then estimate $\norm{\Lambda_0\ket{w_0'}}^2$ to relative accuracy $\eps/4$ using amplitude estimation. Since $\Delta\leq \Delta(U(P,x))$, we have $\norm{\Pi_0^x\ket{w_0}}^2\leq \norm{\Lambda_0\ket{w_0'}}^2\leq \norm{\Pi_\Delta^x\ket{w_0}}^2+\epsilon =\norm{\Pi_0^x\ket{w_0}}^2+\epsilon$. By \lem{pi-0-neg}, we have $\norm{\Pi_0^x\ket{w_0}}^2=\frac{1}{w_-(x)}$, so we will obtain an estimate $\tilde{p}$ of $\frac{1}{w_-(x)}$ such that 
$$\(1-\frac{\eps}{4}\)\frac{1}{w_-(x)}\leq \tilde{p}\leq \(1+\frac{\eps}{4}\)\(\frac{1}{w_-(x)}+\epsilon\). $$
If $\tilde{p}>2(1+\frac{\eps}{4})\epsilon$, then we know that $\frac{1}{w_-(x)}\geq \epsilon$, so we perform one more estimate with accuracy $\epsilon'=\frac{\eps}{8}\epsilon\leq \frac{\eps}{8}\frac{1}{w-(x)}$ and return the resulting estimate. Otherwise, we let $\epsilon'=\epsilon/2$ and repeat. 

To see that we will eventually terminate, suppose $\epsilon \leq\frac{1}{4w_-(x)}$. Then we have 
$$\tilde p\geq (1-\eps/4)\frac{1}{w_-(x)}\geq (3/4)4\epsilon\geq (3/4)({4}/{5})(1+\eps/4) 4\epsilon \geq 2(1+\eps/4)\epsilon,$$
so the algorithm terminates. 
Upon termination, we have 
$$\tilde p \leq \(1+\eps/4\)\(\frac{1}{w_-(x)}+\epsilon\)\leq \(1+\eps/4\)\(\frac{1}{w_-(x)}+\frac{\eps}{8}\frac{1}{w_-(x)}\)\leq \(1+\frac{\eps}{2}\)\frac{1}{w_-(x)},$$ 
so $|1/\tilde p-w_-(x)|\leq \eps w_-(x).$
By \thmthm{phase-est}{amp-est}, the total number of calls to $U$ is:
$$\sum_{i=0}^{\log 4w_-(x)}\frac{1}{\Delta}\frac{\sqrt{w_-(x)}}{\eps}\log 2^i+\frac{\sqrt{w_-(x)}}{\Delta\eps}\log\frac{w_-(x)}{\eps}=\frac{1}{\Delta}\frac{\sqrt{w_-(x)}}{\eps}\(\sum_{i=0}^{\log 6w_-(x)}i+\log\frac{w_-(x)}{\eps}\),$$
which is at most $\frac{\sqrt{w_-(x)}}{\Delta}{\eps}\log^2\frac{w_-(x)}{\eps}=\tO\(\frac{\sqrt{w_-(x)}}{\Delta\eps}\)$.
Similarly, we can estimate $w_+(x)$ to relative accuracy $\eps$ using $\tO\(\frac{\sqrt{w_+(x)}}{\Delta'\eps}\)$ calls to $U'$.
\end{proof}

\thm{gap-improvement} is only useful if a lower bound on the phase gap of $U(P,x)$ or $U'(P,x)$ can be computed. This may not always be feasible, but the following two theorems shows it is sufficient to compute the spectral norm of $A$, and the spectral gap, or specifically, smallest nonzero singular value, of the matrix $A(x)=A\Pi_{H(x)}$. This may still not be an easy task, but in \sec{app}, we show that we can get a better algorithm for estimating the effective resistance by this analysis, which, in the case of effective resistance, is very simple.

\begin{theorem}\label{thm:kappa}
Let $P$ be any span program on $[q]^n$. 
 For any $x\in [q]^n$, $\Delta(U(P,x))\geq 2\frac{\sigma_{\min}(A(x))}{\sigma_{\max}(A)}.$
\end{theorem}
\begin{proof}
Let $U=U(P,x)$. Consider $-U=(2\Pi_{(\ker A)^\bot}-I)(2\Pi_{H(x)}-I)$. By \cor{szegedy}, if $D$ is the discriminant of $-U$, then $\Delta(U)\geq 2\sigma_{\min}(D)$, so we will lower bound $\sigma_{\min}(D)$. 
Since the orthogonal projector onto $(\ker A)^\bot=\mathrm{row}A$ is $A^+A$, we have $D=A^+A\Pi_{H(x)}=A^+A(x)$. 

We have $\sigma_{\min}(D)=\min_{\ket{u}\in\mathrm{row}D}\frac{\norm{D\ket{u}}}{\norm{\ket{u}}}$, so let $\ket{u}\in\mathrm{row}D$ be a unit vector that minimizes $\norm{D\ket{u}}$. Since $\ket{u}\in\mathrm{row}D\subseteq \mathrm{row}A(x)$, we have $\norm{A(x)\ket{u}}\geq \sigma_{\min}(A(x))$. Since $A(x)\ket{u}\in \mathrm{col}A(x)\subseteq \mathrm{col}A=\mathrm{row}A^+$, we have 
$$\sigma_{\min}(D)=\norm{A^+A(x)\ket{u}}\geq \sigma_{\min}(A^+)\norm{A(x)\ket{u}}\geq \sigma_{\min}(A^+)\sigma_{\min}(A(x))=\frac{\sigma_{\min}(A(x))}{\sigma_{\max}(A)},$$
since $\sigma_{\min}(A^+)=\frac{1}{\sigma_{\max}(A)}$. 
Thus $\Delta(U)\geq 2\frac{\sigma_{\min}(A(x))}{\sigma_{\max}(A)}$.
\end{proof}

\begin{theorem}\label{thm:kappa-pos}
Let $P$ be any span program. For any $x\in P_1$, $\Delta(U'(P,x))\geq 2\frac{\sigma_{\min}(A(x))}{\sigma_{\max}(A)}$. 
\end{theorem}
\begin{proof}
We have 
$$
-U'(P,x)^\dagger = (2(I-\Pi_{\ker A\oplus\mathrm{span}\{\ket{w_0}\}})-I)(2\Pi_{H(x)}-I)
 = (2(I-\Pi_{\ker A}-\Pi_{\ket{w_0}})-I)(2\Pi_{H(x)}-I),
$$
since $\ket{w_0}\in (\ker A)^\bot$, so $-U'(P,x)^\dagger$ has discriminant:
$$D'=(\Pi_{(\ker A)^\bot}-\Pi_{\ket{w_0}})\Pi_{H(x)}=\Pi_{(\ker A)^\bot}\Pi_{H(x)}-\Pi_{\ket{w_0}}\Pi_{(\ker A)^\bot}\Pi_{H(x)}
=\Pi_{\ket{w_0}^\bot}D.$$

Since $x\in P_1$, let $\ket{w_x}=A(x)^+\ket{\tau}$. Then $D\ket{w_x}=A^+A(x)\ket{w_x}=A^+\ket{\tau}=\ket{w_0}$, so $\ket{w_0}\in\mathrm{col}D$. Let $\{\ket{\phi_0}=\ket{w_0},\ket{\phi_1},\dots,\ket{\phi_{r-1}}\}$ be an orthogonal basis for $\mathrm{col}D$. Then we can write $D=\sum_{i=0}^{r-1}\ket{\phi_i}\bra{v_i}$ for $\ket{v_i}=D^\dagger\ket{\phi_i}\neq 0$ (not necessarily orthogonal). Then $D'=\sum_{i=0}^{r-1}\Pi_{\ket{w_0}^\bot}\ket{\phi_i}\bra{v_i}=\sum_{i=1}^{r-1}\ket{\phi_i}\bra{v_i}$, so 
$\mathrm{col}D'=\mathrm{span}\{\ket{\phi_1},\dots,\ket{\phi_{r-1}}\}=\{\ket{\phi}\in\mathrm{col}D:\braket{\phi}{w_0}=0\}$. Thus:
\begin{eqnarray*}
\sigma_{\min}(D')&=& \min_{\ket{u}\in\mathrm{col}D'}\frac{\norm{\bra{u}D'}}{\norm{\ket{u}}}
\;\; =\;\; \min_{\ket{u}\in\mathrm{col}D:\braket{w_0}{u}=0}\frac{\norm{\bra{u}\Pi_{\ket{w_0}^\bot}D}}{\norm{\ket{u}}}
\;\; = \;\; \min_{\ket{u}\in\mathrm{col}D:\braket{w_0}{u}=0}\frac{\norm{\bra{u}D}}{\norm{\ket{u}}}\\
&\geq & \min_{\ket{u}\in\mathrm{col}D}\frac{\norm{\bra{u}D}}{\norm{\ket{u}}}
\;\; = \;\; \sigma_{\min}(D).
\end{eqnarray*}
By the proof of \thm{kappa}, we have $\sigma_{\min}(D)\geq \frac{\sigma_{\min}(A(x))}{\sigma_{\max}(A)}$ and by \cor{szegedy}, we have $\Delta(U'(P,x)^\dagger)=\Delta(U'(P,x))\geq 2\sigma_{\min}(D')\geq 2\sigma_{\min}(D)\geq 2\frac{\sigma_{\min}(A(x))}{\sigma_{\max}(A)}$.
\end{proof}

\noindent Combining the last three theorems, we get the following, which has \thm{est-alg-gap} as a special case:
\begin{theorem}\label{thm:est-alg-gap-N}
Fix $f:X\subseteq [q]^n\rightarrow\mathbb{R}_{>0}$, and define $\kappa(f)=\max_{x\in X}\frac{\sigma_{\max}(A)}{\sigma_{\min}(A(x))}$. Let $P$ be any span program on $[q]^n$ such that $X\subseteq P_0$ (resp. $X\subseteq P_1$), and for all $x\in X$, $f(x)=w_-(x,P)$ (resp. $f(x)=w_+(x,P)$). Let $N=\norm{\ket{w_0}}^2$. Then there is a quantum algorithm that estimates $f$ to relative accuracy $\eps$ using $\tO\(\frac{\kappa(f)}{\eps}\sqrt{Nf(x)}\)$ (resp. $\tO\(\frac{\kappa(f)}{\eps}\sqrt{\frac{f(x)}{N}}\)$) queries.  
\end{theorem}
\begin{proof}
Let $P'$ be the span program that is the same as $P$, but with target $\tau'=\frac{\tau}{\sqrt{N}}$. Then it's clear that $\frac{\ket{w_0}}{\sqrt{N}}$ is the minimal positive witness of $P'$, and furthermore, it has norm 1, so $P'$ is normalized. We can similarly see that for any $x\in P_1$, if $\ket{w_x}$ is an optimal positive witness for $x$ in $P$, then $\frac{1}{\sqrt{N}}\ket{w_x}$ is an optimal positive witness for $x$ in $P'$, so $w_+(x,P')=\frac{w_+(x,P)}{N}$. Similarly, for any $x\in P_0$, if $\omega_x$ is an optimal negative witness for $x$ in $P$, then $\sqrt{N}\omega_x$ is an optimal negative witness for $x$ in $P'$, so $w_-(x,P')=Nw_-(x,P)$. By \thmthm{kappa}{kappa-pos}, for all $x\in X$, $\frac{1}{\Delta(U(P',x))}\leq \kappa(f)$ (resp.   
$\frac{1}{\Delta(U'(P',x))}\leq \kappa(f)$). The result then follows from \thm{gap-improvement}.
\end{proof}

\section{Applications}\label{sec:app}

In this section, we will demonstrate how to apply the ideas from \sec{alg} to get new quantum algorithms. 
Specifically, we will give upper bounds of $\tO(n\sqrt{R_{s,t}}/\eps^{3/2})$ and $\tO(n\sqrt{R_{s,t}/\lambda_2}/\eps)$
 on the time complexity of estimating the effective resistance, $R_{s,t}$, between two vertices, $s$ and $t$, in a graph. Unlike previous upper bounds, we study this problem in the adjacency model, however, there are similarities between the ideas of this upper bound and a previous quantum upper bound in the edge-list model due to Wang \cite{Wan13}, which we discuss further at the end of this section.

A \emph{unit flow} from $s$ to $t$ in $G$ is a real-valued function $\theta$ on the directed edges $\overarrow{E}(G)=\{(u,v):\{u,v\}\in E(G)\}$ such that:
\begin{enumerate}
\item for all $(u,v)\in \overarrow{E}$, $\theta(u,v)=-\theta(v,u)$;
\item for all $u\in [n]\setminus\{s,t\}$, $\sum_{v\in \Gamma(u)}\theta(u,v)=0$, where $\Gamma(u)=\{v\in [n]:\{u,v\}\in E\}$; and
\item $\sum_{u\in\Gamma(s)}\theta(s,u)=\sum_{u\in\Gamma(t)}\theta(u,t)=1$.
\end{enumerate}
Let $\cal F$ be the set of unit flows from $s$ to $t$ in $G$. The \emph{effective resistance} from $s$ to $t$ in $G$ is defined:
$$R_{s,t}(G)=\min_{\theta\in {\cal F}}\sum_{\{u,v\}\in E(G)}\theta(u,v)^2.$$

In the adjacency model, we are given, as input, a string $x\in\{0,1\}^{n\times n}$, representing a graph $G_x=([n],\{\{i,j\}:x_{i,j}=1\})$ (we assume that $x_{i,i}=0$ for all $i$, and $x_{i,j}=x_{j,i}$ for all $i,j$). 
The problem of $st$-connectivity is the following. Given as input $x\in\{0,1\}^{n\times n}$ and $s,t\in [n]$, decide if there exists a path from $s$ to $t$ in $G_x$; that is, whether or not $s$ and $t$ are in the same component of $G_x$. A span-program-based algorithm for this problem was given in \cite{BR12}, with time complexity $\tO(n\sqrt{p})$, under the promise that, if $s$ and $t$ are connected in $G_x$, they are connected by a path of length $\leq p$.  They use the following span program, defined on $\{0,1\}^{n\times n}$:
$$H_{(u,v),0}=\{0\},\;\; H_{(u,v),1}=\mathrm{span}\{\ket{u,v}\},\;\; V=\mathbb{R}^{n},\;\; A=\sum_{u,v\in [n]}(\ket{u}-\ket{v})\bra{u,v},\;\; \ket{\tau}=\ket{s}-\ket{t}.$$
We have $H=\mathrm{span}\{\ket{u,v}:u,v\in [n]\}$, and $H(x)=\mathrm{span}\{\ket{u,v}:\{u,v\}\in E(G_x)\}$.  Throughout this section, $P$ will denote the above span program. We will use this span program to define algorithms for estimating the effective resistance. Ref.\ \cite{BR12} are even able to show how to efficiently implement a unitary similar to $U(P,x)$, giving a time efficient algorithm. In \app{time}, we adapt their proof to our setting, showing how to efficiently implement $U'(P^\beta,x)$ for any $n^{-O(1)}\leq \beta\leq n^{O(1)}$ and efficiently construct the initial state $\ket{w_0}$, making our algorithms time efficient as well.

The effective resistance between $s$ and $t$ is related to $st$-connectivity by the fact that if $s$ and $t$ are not connected, then $R_{s,t}$ is undefined (there is no flow from $s$ to $t$) and if $s$ and $t$ are connected then $R_{s,t}$ is related to the number and length of paths from $s$ to $t$. In particular, if $s$ and $t$ are connected by a path of length $p$, then $R_{s,t}(G)\leq p$ (take the unit flow that simply travels along this path). In general, if $s$ and $t$ are connected in $G$, then $\frac{2}{n}\leq R_{s,t}(G)\leq n-1$. The span program for $st$-connectivity is amenable to the task of estimating the effective resistance due to the following.
\begin{lemma}[\cite{BR12}]\label{lem:st-pos-witness}
For any graph $G_x$ on $[n]$, $x\in P_1$ if and only if $s$ and $t$ are connected, and in that case, $w_+(x,P)=\frac{1}{2}R_{s,t}(G_x)$. 
\end{lemma}
\noindent A near immediate consequence of this, combined with \thm{est-alg}, is the following.
\begin{theorem}\label{thm:resistance-effective-gap}
There exists a quantum algorithm for estimating $R_{s,t}(G_x)$ to accuracy $\eps$ with time complexity $\tO\(\frac{n\sqrt{R_{s,t}(G_x)}}{\eps^{3/2}}\)$ and space complexity $O(\log n)$. 
\end{theorem}
\begin{proof}
We merely observe that if $G$ is a connected graph, an approximate negative witness is $\omega:[n]\rightarrow\mathbb{R}$ that minimizes $\norm{\omega A\Pi_{H(x)}}^2=\sum_{\{u,v\}\in E}(\omega(u)-\omega(v))^2$ and satisfies $\omega(s)-\omega(t)=1$. That is, $\omega$ is the voltage induced by a unit potential difference between $s$ and $t$ (see \cite{DS84} for details). This is not unique, but if we fix $\omega(s)=1$ and $\omega(t)=0$, then the $\omega$ that minimizes $\norm{\omega A\Pi_{H(x)}}^2$ is unique, and this is without loss of generality. In that case, for all $u\in [n]$, $0\leq \omega(u)\leq 1$, so 
$$\textstyle\tilde{w}_-(x)=\norm{\omega A}^2=\sum_{u,v\in [n]}(\omega(u)-\omega(v))\leq 2n^2\qquad\mbox{ and thus }\qquad \widetilde{W}_-\leq 2n^2.$$
By \thm{est-alg}, we can estimate $R_{s,t}$ to precision $\eps$ using $\tO\(\frac{\sqrt{\widetilde{W}_-w_+(x)}}{\eps^{3/2}}\)=\tO\(\frac{n\sqrt{R_{s,t}(G_x)}}{\eps^{3/2}}\)$ calls to $U'(P^\beta,x)$ for some $\beta$, which, by \thm{time}, costs $O(\log n)$ time and space. 
\end{proof}

By analyzing the spectra of $A$ and $A(x)$, and applying \thm{est-alg-gap}, we can get an often better algorithm (\thm{resistance-real-gap}). The \emph{spectral gap} of a graph $G$, denoted $\lambda_2(G)$, is the second largest eigenvalue (including multiplicity) of the Laplacian of $G$, which is defined $L_G=\sum_{u\in [n]}d_u\ket{u}\bra{u}-\sum_{u\in [n]}\sum_{v\in\Gamma(u)}\ket{u}\bra{v}$, where $d_u$ is the degree of $u$, and $\Gamma(u)$ is the set of neighbours of $u$. 
The smallest eigenvalue of $L_G$ is $0$ for any graph $G$. A graph $G$ is connected if and only if $\lambda_2(G)>0$. A connected graph $G$ has $\frac{2}{n^2}\leq \lambda_2(G)\leq n$.

The following theorem is an improvement over \thm{resistance-effective-gap} when $\lambda_2(G)> \eps$. In particular, it is an improvement for all $\eps$ when we know that $\lambda_2(G)>1$.

\begin{theorem}\label{thm:resistance-real-gap}
Let $\mathcal{G}$ be a family of graphs such that for all $x\in \mathcal{G}$, $\lambda_2(G_x)\geq \mu$. Let $f:\mathcal{G}\times [n]\times [n]\rightarrow \mathbb{R}_{>0}$ be defined by $f(x,s,t)=R_{s,t}(G_x)$. There exists a quantum algorithm for estimating $f$ to relative accuracy $\eps$ that has time complexity $\tO\(\frac{1}{\eps}n\sqrt{R_{s,t}(G_x)/\mu}\)$ and space complexity $O(\log n)$.
\end{theorem}
\begin{proof}
We will apply \thm{est-alg-gap}. We first compute $\norm{\ket{w_0}}^2$, in order to normalize $P$.
\begin{lemma}\label{lem:w0}$N=\norm{\ket{w_0}}^2=\frac{1}{n}$.
\end{lemma}
\begin{proof}
Since $H(x)=H$ when $G_x$ is the complete graph, by \lem{st-pos-witness}, we need only compute $R_{s,t}$ in the complete graph. It's simple to verify that the optimal unit $st$-flow in the complete graph has $\frac{1}{n}$ units of flow on every path of the form $(s,u,t)$ for $u\in [n]\setminus\{s,t\}$, and $\frac{2}{n}$ units of flow on the edge $(s,t)$. Thus, $R_{s,t}(K_n)=\sum_{u\in [n]\setminus \{s,t\}}2(1/n)^2+(2/n)^2=2/n$. Thus $\norm{\ket{w_0}}^2=\frac{1}{2}R_{s,t}(K_n)=\frac{1}{n}$.
\end{proof}
\noindent Next, we compute the following:
\begin{lemma}\label{lem:st-conn-phase}For any $x\in\mathcal{G}$, $\frac{\sigma_{\max}(A)}{\sigma_{\min}(A(x))}=\sqrt{\frac{n}{\lambda_2(G_x)}}\leq \sqrt{\frac{n}{\mu}}$, so $\kappa(f)\leq \sqrt{\frac{n}{\mu}}$.
\end{lemma}
\begin{proof}
Let $L_x$ denote the Laplacian of $G_x$. 
We have: 
$$A(x)A(x)^T = \sum_{u\in [n]}\sum_{v\in\Gamma(u)}(\ket{u}-\ket{v})(\bra{u}-\bra{v})=2\sum_{u\in [n]}d_u\ket{u}\bra{u}-2\sum_{u\in [n]}\sum_{v\in\Gamma(u)}\ket{u}\bra{v}=2L_x.$$
Thus, if $L$ denotes the Laplacian of the complete graph, we also have $AA^T=2L$.
Letting $J$ denote the all ones matrix, we have $L=(n-1)I-(J-I)=nI-J$, and since $J=n\ket{u}\bra{u}$ where $\ket{u}=\frac{1}{\sqrt{n}}\sum_{i=1}^n\ket{i}$, if $\ket{u_1},\dots,\ket{u_{n-1}},\ket{u}$ is any orthonormal basis of $\mathbb{R}^n$, then $L=n\sum_{i=1}^{n-1}\ket{u_i}\bra{u_i}+n\ket{u}\bra{u}-n\ket{u}\bra{u}=\sum_{i=1}^{n-1}n\ket{u_i}\bra{u_i}$, so the spectrum of $L$ is $0$, with multiplicity $1$, and $n$ with multiplicity $n-1$. Thus, the only nonzero singular value of $A$ is $\sqrt{2n}=\sigma_{\max}(A)$.
Furthermore, since $\lambda_2(G_x)$ is the smallest nonzero eigenvalue of $L_x$, and $A(x)A(x)^T=2L_x$, $\sigma_{\min}(A(x))=\sqrt{2\lambda_2(G_x)}$. 
The result follows.
\end{proof}
Finally, by \lem{st-pos-witness}, we have $w_+(x,P)=\frac{1}{2}R_{s,t}(G_x)$, so, applying \thm{est-alg-gap-N}, we get an algorithm that makes $\tO\(\frac{\kappa(f)}{\eps}\sqrt{\frac{w_+(x,P)}{N}}\)=\tO\(\frac{1}{\eps}\sqrt{n/\mu}\sqrt{R_{s,t}n}\)$ calls to $U'(P,x)$. By \thm{time}, this algorithm has time complexity $\tO\(\frac{1}{\eps}n\sqrt{R_{s,t}/\mu}\)$ and space complexity $O(\log n)$.
\end{proof}

Both of our upper bounds have linear dependence on $n$, and the following theorem shows that this is optimal.
\begin{theorem}[Lower Bound]
There exists a family of graphs $\mathcal{G}$ such that estimating effective resistance on $\mathcal{G}$ costs at least $\Omega(n)$ queries. 
\end{theorem}
\begin{proof}
Let $\mathcal{G}_0$ be the set of graphs consisting of two stars $K_{1,n/2-1}$, centered at $s$ and $t$, with an edge connecting $s$ and $t$ (see \fig{st-lb}). Let $\mathcal{G}_1$ be the set of graphs consisting of graphs from $\mathcal{G}_0$ with a single edge added between two degree one vertices from different stars. Let $\mathcal{G}=\mathcal{G}_0\cup\mathcal{G}_1$. We first note that we can distinguish between $\mathcal{G}_0$ and $\mathcal{G}_1$ by estimating effective resistance on $\mathcal{G}$ to accuracy $\frac{1}{10}$: If $G\in\mathcal{G}_0$, then there is a single $st$-path, consisting of one edge, so the effective resistance is $1$. If $G\in\mathcal{G}_1$, then there are two $st$-paths, one of length 1 and one of length 3. We put a flow of $\frac{1}{4}$ on the length-3 path and $\frac{3}{4}$ on the length-1 path to get effective resistance at most $(3/4)^2+3(1/4)^2=\frac{3}{4}$. 

\begin{figure}
\centering
\begin{tikzpicture}[scale=1.2]
\draw[dashed,color=gray] (1,.75)--(2,.25);
\draw[dashed,color=gray] (1,.75)--(2,.5);
\draw[dashed,color=gray] (1,.75)--(2,.75);
\draw[dashed,color=gray] (1,.5)--(2,.25);
\draw[dashed,color=gray] (1,.5)--(2,.5);
\draw[dashed,color=gray] (1,.5)--(2,.75);
\draw[dashed,color=gray] (1,.25)--(2,.25);
\draw[dashed,color=gray] (1,.25)--(2,.5);
\draw[dashed,color=gray] (1,.25)--(2,.75);

\filldraw (0,0) circle (.05);
\filldraw (3,0) circle (.05);
\draw (0,0)--(3,0);

\filldraw (1,.25) circle (.05);
\filldraw (1,.5) circle (.05);
\filldraw (1,.75) circle (.05);

\filldraw (2,.25) circle (.05);
\filldraw (2,.5) circle (.05);
\filldraw (2,.75) circle (.05);

\draw (0,0) -- (1,.25);
\draw (0,0) -- (1,.5);
\draw (0,0) -- (1,.75);

\draw (3,0) -- (2,.25);
\draw (3,0) -- (2,.5);
\draw (3,0) -- (2,.75);

\node at (-.25,0) {$s$};
\node at (3.25,0) {$t$};
\end{tikzpicture}
\caption{The graphs in $\mathcal{G}_0$ contain only the solid edges. The graphs in $\mathcal{G}_1$ contain the solid edges and one of the dashed edges. We can embed an instance of OR in the dashed edges. If one of the dashed edges is included, the number of $st$-paths increases to 2, decreasing the effective resistance.}\label{fig:st-lb}
\end{figure}
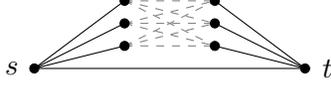

We now describe how to embed an instance $y\in \{0,1\}^{(n/2-1)^2}$ of OR$_{(n/2-1)^2}$ in a graph. We let $s=1$ be connected to every vertex in $\{2,\dots,n/2\}$, and $t=n$ be connected to every vertex in $\{n/2+1,\dots,n-1\}$. Let the values of $\{G_{i,j}:i\in \{2,\dots,n/2\},j\in \{n/2,\dots,n-1\}\}$ be determined by $y$. Let all other values $G_{i,j}$ be 0. Then clearly $R_{s,t}(G)\geq 1$ if and only if $y=0\dots 0$ (in that case $G\in \mathcal{G}_0$) and otherwise, $R_{s,t}(G)\leq 3/4$, since there is at least one extra path from $s$ to $t$ (in that case $G\in \mathcal{G}_1$). 
The result follows from the lower bound of $\Omega(\sqrt{(n/2-1)^2})=\Omega(n)$ on OR$_{(n/2-1)^2}$.
\end{proof}

\paragraph{Discussion} The algorithms from \thmthm{resistance-effective-gap}{resistance-real-gap} are the first quantum algorithms for estimating the effective resistance in the adjacency model, however, the problem has been studied previously in the edge-list model \cite{Wan13}, where Wang obtains a quantum algorithm with complexity $\tO\(\frac{d^{3/2}\log n}{\Phi(G)^2\eps}\)$, where $\Phi(G)\leq 1$ is the conductance (or edge-expansion) of $G$. In the edge-list model, the input $x\in [n]^{[n]\times [d]}$ models a $d$-regular graph (or $d$-bounded degree graph) $G_x$ by $x_{u,i}=v$ for some $i\in [d]$ whenever $\{u,v\}\in E(G_x)$.
Wang requires edge-list queries to simulate walking on the graph, which requires constructing a superposition over all neighbours of a given vertex. This type of edge-list query can be simulated by $\sqrt{n/d}$ adjacency queries to a $d$-regular graph, using quantum search, so Wang's algorithm can be converted to an algorithm in the adjacency query model with cost $\tO\(\frac{d^{3/2}}{\Phi(G)^2\eps}\sqrt{\frac{n}{d}}\)$. We can compare our results to this by noticing that $R_{s,t}\leq \frac{1}{\lambda_2(G)}$ \cite{CRRST96}, implying that our algorithm always runs in time at most $\tO\(\frac{1}{\eps}\frac{n}{\mu}\)$. 
If $G$ is a connected $d$-regular graph, then $\lambda_2(G)=d\delta(G)$, where $\delta(G)$ is the spectral gap of a random walk on $G$. By Cheeger inequalities, we have $\frac{\Phi^2}{2}\leq \delta$ \cite{LPW09}, so the complexity of the algorithm from \thm{resistance-real-gap} is at most $\tO\(\frac{1}{\eps}\frac{n}{d\delta}\)=\tO\(\frac{1}{\eps}\frac{n}{d\Phi^2}\)$, which is an improvement over the bound of $\tO\(\frac{1}{\eps}\frac{d^{3/2}}{\Phi^2}\sqrt{\frac{n}{d}}\)=\tO\(\frac{1}{\eps}\frac{d}{\Phi^2}\sqrt{n}\)$ given by naively adapting Wang's algorithm to the adjacency model whenever $d>\sqrt[4]{n}$. In general our upper bound may be much better than $\frac{1}{\eps}\frac{n}{d\Phi^2}$, since the Cheeger inequality is not tight, and $R_{s,t}$ can be much smaller than $\frac{1}{\lambda_2}$. 

It is worth further discussing Wang's algorithms for estimating effective resistance, due to their relationship with the ideas presented here. In order to get a time-efficient algorithm for $st$-connectivity, Belovs and Reichardt show how to efficiently reflect about the kernel of $A$ (see also \app{time}), $A$ being related to the Laplacian of a complete graph, $L$, by $AA^T=2L$. This implementation consists, in part, of a quantum walk on the complete graph. Wang's algorithm directly implements a reflection about the kernel of $A(x)$ by instead using a quantum walk on the graph $G$, which can be done efficiently in the edge-list model. For general span programs, when a reflection about the kernel of $A(x)$ can be implemented efficiently in such a direct way, this can lead to an efficient quantum algorithm for estimating the witness size. 

We also remark on another quantum algorithm for estimating effective resistance, also from \cite{Wan13}. This algorithm has the worse complexity $\tO\(\frac{d^8\mathrm{polylog}n}{\Phi(G)^{10}\eps^2}\)$, and is obtained by using the HHL algorithm \cite{HHL09} to estimate $\norm{A(x)^+\ket{\tau}}^2$, which is the positive witness size of $x$, or in this case, the effective resistance. We remark that, for any span program, $w_+(x)=\norm{\ket{w_x}}^2=\norm{A(x)^+\ket{\tau}}^2$, so HHL may be another means of estimating the positive witness size. There are several caveats: $A(x)$ must be efficiently row-computable, and the complexity additionally depends on $\frac{\sigma_{\max}(A(x))}{\sigma_{\min}(A(x))}$, the \emph{condition number} of $A(x)$ (We remark that this is upper bounded by $\frac{\sigma_{\max}(A)}{\sigma_{\min}(A(x))}$, upon which the complexity of some of our algorithms depends as well). However, if this approach yields an efficient algorithm, it is efficient in time complexity, not only query complexity. We leave further exploration of this idea for future research.

\section{Conclusion and Open Problems}\label{sec:open}

\paragraph{Summary} We have presented several new techniques for turning span programs into quantum algorithms, which we hope will have future applications. Specifically, given a span program $P$, in addition to algorithms for deciding any function $f$ such that $f^{-1}(0)\subseteq P_0$ and $f^{-1}(1)\subseteq P_1$, we also show how to get several different algorithms for deciding a number of related threshold problems, as well as estimating the witness size. In addition to algorithms based on the standard effective spectral gap lemma, we also show how to get algorithms by analyzing the real phase gap.

We hope that the importance of this work lies not only in its potential for applications, but in the improved understanding of the structure and  power of span programs. A number of very important quantum algorithms rely on a similar structure, using phase estimation of a unitary that depends on the input to distinguish between different types of inputs. Span-program-based algorithms represent a very general class of such algorithms, making them not only important to the study of the quantum query model, but to quantum algorithms in general. 

\paragraph{Further Applications} The main avenue for future work is in applications of our techniques to obtain new quantum algorithms. We stress that \emph{any} span program for a decision problem can now be turned into an algorithm for estimating the positive or negative witness size, if these correspond to some meaningful function, or deciding threshold functions related to the witness size. A natural source of potential future applications is in the rich area of property testing problems (for a survey, see \cite{MdW13}).

\paragraph{Span Programs and HHL} 
One final open problem, briefly discussed at the end of the previous section, is the relationship between estimating the witness size and the HHL algorithm  \cite{HHL09}. The HHL algorithm can be used to estimate $\norm{M^+\ket{u}}^2$, given the state $\ket{u}$ and access to a row-computable linear operator $M$. When $M=A(x)$, this quantity is exactly $w_+(x)$, so if $A(x)$ is row-computable --- that is, there is an efficient procedure for computing the \th{$i$} nonzero entry of the \th{$j$} row of $A(x)$, then HHL gives us yet another means of estimating the witness size, whose time complexity is known, rather than only its query complexity. 
It may be interesting to explore this connection further.

\section{Acknowledgements}

The authors would like to thank David Gosset, Shelby Kimmel, Ben Reichardt, and Guoming Wang
for useful discussions about span programs. 
We would especially like to thank Shelby Kimmel for valuable feedback and suggestions on an earlier draft of this paper. 
Finally, S.J.\ would like to thank Moritz Ernst for acting as a sounding board throughout the writing of this paper.

{\small 
\bibliographystyle{alpha}
\bibliography{refs}
}

\appendix

\section{Span Program Scaling}\label{app:scaling}

In this section we prove \thm{scaling}. Let $P=(H,V,\tau,A)$ be any span program on $[q]^n$, and let $N=\norm{\ket{w_0}}^2$ for $\ket{w_0}$ the optimal positive witness of $P$. We define $P^\beta=(H^{\beta},A^{\beta},\tau^{\beta},V^{\beta})$ as follows. Let $\ket{\hat{0}}$ and $\ket{\hat{1}}$ be two vectors orthogonal to $H$ and $V$. We define:
$$
\forall j\in [n], a\in [q], H_{j,a}^{\beta}=H_{j,a},\quad
H^{\beta}_{\mathrm{true}}=H_{\mathrm{true}}\oplus \mathrm{span}\{\ket{\hat 1}\},\quad
H^{\beta}_{\mathrm{false}}=H_{\mathrm{false}}\oplus \mathrm{span}\{\ket{\hat 0}\}$$
$$V^{\beta}=V\oplus\mathrm{span}\{\ket{\hat 1}\},\quad 
A^{\beta}=\beta A+\tau\ket{\hat 0}+\frac{\sqrt{\beta^2+N}}{\beta}\ket{\hat 1}\bra{\hat 1},\quad
\tau^{\beta}=\tau+\ket{\hat 1}$$
We then have and $H^{\beta}=H\oplus \mathrm{span}\{\ket{\hat 0},\ket{\hat{1}}\}$ and $H^{\beta}(x)=H(x)\oplus \mathrm{span}\{\ket{\hat 1}\}$. 
In order to prove \thm{scaling}, we will show that:
\begin{itemize}
\item For all $x\in P_1$, $w_+(x,P^{\beta})=\frac{1}{\beta^2}w_+(x,P)+\frac{\beta^2}{N+\beta^2}$ and $\tilde{w}_-(x,P^{\beta})\leq \beta^2\tilde{w}_-(x,P)+2$;
\item for all $x\in P_0$, $w_-(x,P^{\beta})=\beta^2 w_-(x,P)+1$ and $\tilde{w}_+(x,P^{\beta})\leq \frac{1}{\beta^2}\tilde{w}_+(x,P)+2$;
\item the smallest witness in $P^{\beta}$ is $\ket{w_0^\beta}=\frac{\beta}{\beta^2+N}\ket{w_0}+\frac{N}{\beta^2+N}\ket{\hat 0}+\frac{\beta}{\sqrt{\beta^2+N}}\ket{\hat 1}$, and $\norm{\ket{w_0^\beta}}^2=1$.
\end{itemize}

\begin{lemma}\label{lem:scaled-w0}
The smallest witness in $P^{\beta}$ is $\ket{w_0^\beta}=\frac{\beta}{\beta^2+N}\ket{w_0}+\frac{N}{\beta^2+N}\ket{\hat 0}+\frac{\beta}{\sqrt{\beta^2+N}}\ket{\hat 1}$. It is easily verified that $\norm{\ket{w_0^\beta}}^2=1$.
\end{lemma}
\begin{proof}
Let $\ket{w_0'}=\ket{h}+b\ket{\hat 0}+c\ket{\hat 1}$ be the smallest witness in $P^{\beta}$, for some $\ket{h}\in H$. Since $A^{\beta}\ket{w_0'}=\beta  A\ket{h}+b\tau+c\frac{\sqrt{\beta^2+N}}{\beta}\ket{\hat 1}=\tau+\ket{\hat 1}$, we must have $c=\frac{\beta}{\sqrt{\beta^2+N}}$ and $A\ket{h}=\frac{1-b}{\beta}\tau$, so $\ket{h}=\frac{1-b}{\beta}\ket{w}$ for some positive witness $\ket{w}$ of $P$. We have:
$$\norm{\ket{w_0'}}^2 = \frac{(1-b)^2}{\beta^2}\norm{\ket{w}}^2+b^2+\frac{\beta^2}{\beta^2+N}.$$
This is minimized by taking $\ket{w}=\ket{w_0}$, the smallest witness of $P$, and setting $b=\frac{N}{\beta^2+N}$, giving:
$$\ket{w_0^\beta} = \frac{\beta}{\beta^2+N}\ket{w_0}+\frac{N}{\beta^2+N}\ket{\hat 0}+\frac{\beta}{\sqrt{\beta^2+N}}\ket{\hat 1}.\qedhere$$
\end{proof}

\begin{lemma}
For all $x\in P_1$, $w_+(x,P^{\beta})=\frac{1}{\beta^2}w_+(x,P)+\frac{\beta^2}{N+\beta^2}$ and $\tilde{w}_-(x,P^\beta)\leq \beta^2\tilde{w}_-(x,P)+2$.
\end{lemma}
\begin{proof}
The proof is similar to that of \lem{scaled-w0}, however, we have $H^{\beta}(x)=H(x)\oplus\mathrm{span}\{\ket{\hat 1}\}$, so a positive witness for $x$ has the form $\ket{w_x'}=\ket{h}+\frac{\beta}{\sqrt{\beta^2+N}}\ket{\hat 1}$ with $\beta\ket{h}$ some witness for $x$ in $P$. Clearly $\norm{\ket{w_x'}}$ is minimized by setting $\ket{h}=\frac{1}{\beta}\ket{w_x}$  for $\ket{w_x}$  the minimal positive witness for $x$ in $P$, so we have $w_+(x,P^{\beta})=\frac{1}{\beta^2}w_+(x,P)+\frac{\beta^2}{\beta^2+N}$, as required. 

Let $\tilde{\omega}$ be an optimal min-error witness for $x$ in $P$, and define
$$\tilde{\omega}'=\frac{(\beta^2+N)w_+(x,P)}{\beta^4+(\beta^2+N)w_+(x,P)}\tilde{\omega}+\frac{\beta^4}{\beta^4+(\beta^2+N)w_+(x,P)}\bra{\hat 1}.$$
We have $\displaystyle\tilde{\omega}'(\tau+\ket{\hat 1})=\frac{(\beta^2+N)w_+(x,P)}{\beta^4+(\beta^2+N)w_+(x,P)}\tilde{\omega}(\tau)+\frac{\beta^4}{\beta^4+(\beta^2+N)w_+(x,P)}=1$, and:
\begin{eqnarray*}
\norm{\tilde{\omega}'A^{\beta}\Pi_{H^\beta(x)}}^2 \!\!\!\! &= & \!\!\!\!\norm{\frac{(\beta^2+N)w_+(x,P)}{\beta^4+(\beta^2+N)w_+(x,P)}\tilde{\omega} \beta A\Pi_{H(x)}}^2+\norm{\frac{\beta^4}{\beta^4+(\beta^2+N)w_+(x,P)}\frac{\sqrt{\beta^2+N}}{\beta}\bra{\hat 1}}^2\\
&=& \!\!\!\!\frac{(\beta^2+N)^2w_+(x,P)^2\beta^2}{(\beta^4+(\beta^2+N)w_+(x,P))^2}\frac{1}{w_+(x,P)}+\frac{\beta^8}{(\beta^4+(\beta^2+N)w_+(x,P))^2}\frac{\beta^2+N}{\beta^2}\\
&=& \!\!\!\!\frac{(\beta^2+N)^2w_+(x,P)\beta^2+\beta^6(\beta^2+N)}{(\beta^4+(\beta^2+N)w_+(x,P))^2}
 =  \frac{\beta^2(\beta^2+N)}{\beta^4+(\beta^2+N)w_+(x,P)}= \frac{1}{w_+(x,P^\beta)}
\end{eqnarray*}
so $\tilde{\omega}'$ is a min-error witness for $x$ in $P^\beta$. Thus, letting $\eps=\frac{(\beta^2+N)w_+(x,P)}{\beta^4+(\beta^2+N)w_+(x,P)}$, we have
\begin{eqnarray*}
\tilde{w}_-(x,P^\beta) &\leq & \norm{\tilde{\omega}' A^\beta}^2 \; = \; \norm{\eps\tilde{\omega} \beta A+\eps\tilde{\omega}(\tau)\bra{\hat 0}+\frac{\sqrt{\beta^2+N}}{\beta}\tilde{\omega}'(\hat 1)\bra{\hat 1}}^2 \\
&\leq & \beta^2\norm{\tilde{\omega} A}^2+1+\frac{\beta^2+N}{\beta^2}\frac{\beta^8}{(\beta^4+(\beta^2+N)w_+(x,P))^2}\\
&\leq & \beta^2\tilde{w}_-(x,P)+1+\frac{\beta^6(\beta^2+N)}{(\beta^4+\beta^2w_+(x,P))^2}\; \leq \; \beta^2\tilde{w}_+(x,P)+2,
\end{eqnarray*}
where in the last line, we use the fact that $w_+(x,P)\geq N$. 
\end{proof}

\begin{lemma}
For all $x\in P_0$, $w_-(x,P^{\beta})=\beta^2 w_-(x,P)+1$, and $\tilde{w}_+(x,P^{\beta})\leq \frac{1}{\beta^2}\tilde{w}_+(x,P)+2$. 
\end{lemma}
\begin{proof}
Let $\omega_x'$ be an optimal negative witness for $x$ in $P^{\beta}$. Since $\omega_x'\Pi_{H^{\beta}(x)}=0$, $\omega_x'\ket{\hat 1}=0$, so $\omega_x'(\tau^{\beta})=\omega_x'(\tau)+\omega_x'(\ket{\hat 1})=\omega_x'(\tau)=1$. Furthermore, $\omega_x'$ minimizes 
$$\norm{\omega_x'A^{\beta}}^2=\norm{\beta\omega_x'A+\omega_x'(\tau)\ket{\hat 0}}^2=\beta^2\norm{\omega_x' A}^2+1.$$
This is minimized by taking $\omega_x'|_V$ to be the minimal negative  witness of $x$ in $P$, so $\norm{\omega_x' A}^2=w_-(x,P)$, and thus $w_-(x,P^{\beta})=\beta^2w_-(x,P)+1$. 

Next, let $\ket{\tilde{w}}$ be an optimal min-error positive witness for $x$ in $P$. Define:
$$\ket{\tilde{w}'}:=\frac{\beta w_-(x,P)}{1+\beta^2w_-(x,P)}\ket{\tilde{w}}+\frac{1}{1+\beta^2w_-(x,P)}\ket{\hat 0}+\frac{\beta}{\sqrt{\beta^2+N}}\ket{\hat 1}.$$
We have:
$$A\ket{\tilde{w}'}=\frac{\beta^2 w_-(x,P)}{1+\beta^2w_-(x,P)}\tau+\frac{1}{1+\beta^2w_-(x,P)}\tau+\ket{\hat 1}=\tau+\ket{\hat 1}=\tau^{\beta},$$
and since $H^{\beta}(x)^\bot = H(x)^\bot\oplus \mathrm{span}\{\ket{\hat 0}\}$:
\begin{eqnarray*}
\norm{\Pi_{H^{\beta}(x)^\bot}\ket{\tilde{w}'}}^2 &=& \norm{\Pi_{H(x)^\bot}\ket{\tilde{w}'}}^2+\norm{\Pi_{\ket{\hat 0}}\ket{\tilde{w}'}}^2\\
&=& \frac{\beta^2w_-(x,P)^2}{(1+\beta^2w_-(x,P))^2}\norm{\Pi_{H(x)^\bot}\ket{\tilde{w}}}^2+\frac{1}{(1+\beta^2w_-(x,P))^2}\\
&=& \frac{\beta^2w_-(x,P)^2}{(1+\beta^2w_-(x,P))^2}\frac{1}{w_-(x,P)}+\frac{1}{(1+\beta^2w_-(x,P))^2}\\
&=&\frac{1}{1+\beta^2w_-(x,P)}\;\;=\;\;\frac{1}{w_-(x,P^\beta)},\\
\end{eqnarray*}
so $\ket{\tilde{w}'}$ has minimal error. Thus:
\begin{eqnarray*}
\tilde{w}_+(x,P^\beta) &\leq& \norm{\ket{\tilde{w}'}}^2\;\; = \;\;\frac{\beta^2w_-(x,P)^2}{(1+\beta^2w_-(x,P))^2}\norm{\ket{\tilde{w}}}^2+\frac{1}{(1+\beta^2w_-(x,P))^2}+\frac{\beta^2}{\beta^2+N}\\
&\leq & \frac{\beta^2w_-(x,P)^2\tilde{w}_+(x,P)}{(1+\beta^2w_-(x,P))^2}+2\;\;\leq \;\; \frac{\beta^2w_-(x,P)^2\tilde{w}_+(x,P)}{\beta^4w_-(x,P)^2}+2
\;\;= \;\; \frac{\tilde{w}_+(x,P)}{\beta^2}+2. \qedhere
\end{eqnarray*}
\end{proof}

\section{Time Complexity Analysis}\label{app:time}

In \cite{BR12}, the authors analyze the time complexity of the reflections needed to implement their span program to give a time upper bound on $st$-connectivity. Since our algorithms look superficially different from theirs, we reproduce their analysis here to show an upper bound on the quantum time complexity of estimating effective resistance.

\begin{theorem}\label{thm:time}
Let $P$ be the span program for $st$-connectivity given in \sec{app}. Then for any $\beta$ such that $1/n^{O(1)}\leq \beta\leq n^{O(1)}$, $U'(P^{\beta},x)$ can be implemented in quantum time complexity $O(\log n)$ and space $O(\log n)$, and $\ket{w_0^\beta}$ can be constructed in quantum time complexity $O(\log n)$. 
\end{theorem}
\begin{proof}
In order to implement $U'(P^\beta,x)$, we must implement the reflections $R_x(\beta)=2\Pi_{H^\beta(x)}-I$ and $R_P'(\beta)=2\Pi_{\ker A^\beta\oplus \mathrm{span}\{\ket{w_0^\beta}\}}-I$.
We remark that $R_x(\beta)$ is easily implemented in a single query and constant overhead. This proof deals with the implementation of $R_P'(\beta)$, which can be easily implemented given an implementation of $R_P=2\Pi_{\ker A}-I$. 

In order to implement $R_P$, we describe a unitary $W=(2\Pi_Z-I)(2\Pi_Y-I)$ that can be efficiently implemented, and such that $W$ can be used to implement $R_P$. In order to show that $W$ implements $R_P$, we need to show that some isometry $M_Y:H\rightarrow Y$ maps $\ker A$ to the $-1$-eigenspace of $W$, and $(\ker A)^\bot$ to the $1$-eigenspace of $W$. This allows us to implement $R_P$ by first implementing the isometry $M_Y$, applying $W$, and then uncomputing $M_Y$. 

Define the spaces $Z$ and $Y$ as follows:
$$Z=\mathrm{span}\left\{\ket{z_u}:=\frac{1}{\sqrt{2(n-1)}}\sum_{v\neq u}\ket{0,u,u,v}+\frac{1}{\sqrt{2(n-1)}}\sum_{v\neq u}\ket{1,u,v,u}:u\in [n]\right\}; \quad\mbox{and}$$
$$Y=\mathrm{span}\left\{\ket{y_{u,v}}:=\(\ket{0,u,u,v}-\ket{1,v,u,v}\)/\sqrt{2}:u,v\in [n], u\neq v\right\}.$$
Define isometries
$$M_Z=\sum_{u\in [n]}\ket{z_u}\bra{u}\quad\mbox{ and }\quad M_Y=\sum_{(u,v)\in [n]^2:u\neq v}\ket{y_{u,v}}\bra{u,v}.$$

\begin{lemma}Let $S=\{M_Y\ket{\psi}:\ket{\psi}\in \ker A\}$ and $S'=\{M_Y\ket{\psi}:\ket{\psi}\in(\ker A)^\bot\}$ be the images of $\ker A$ and $(\ker A)^\bot$ respectively under the isometry $M_Y$. Then $S=Y\cap Z^\bot$, which is exactly the intersection of $Y$ and the $-1$-eigenspace of $W$, and $S'=Y\cap Z$, which is exactly the intersection of $Y$ and the $1$-eigenspace of $W$.
\end{lemma}
\begin{proof}
We have:
\begin{eqnarray*}
M_Z^\dagger M_Y &=& \frac{1}{2\sqrt{n-1}}\sum_{u\in [n]}\sum_{v\neq u}\ket{u}\(\bra{0,u,u,v}+\bra{1,u,v,u}\)\sum_{a,b\in [n]:a\neq b}(\ket{0,a,a,b}-\ket{1,b,a,b})\bra{a,b}\\
&=& \frac{1}{2\sqrt{n-1}}\sum_{u\in [n]}\sum_{v\neq u}\ket{u}\bra{u,v}-\frac{1}{2\sqrt{n-1}}\sum_{u\in [n]}\sum_{v\neq u}\ket{v}\bra{u,v}\;=\;\frac{1}{2\sqrt{n-1}}A.
\end{eqnarray*}
Thus, for all $\ket{\psi}\in \ker A$, $M_Y\ket{\psi}\in Y\cap \ker M_Z^\dagger=Y\cap Z^\bot$, so $S\subseteq Y\cap Z^\bot$. On the other hand, if $\ket{\psi}\in (\ker A)^\bot$, then $M_Y\ket{\psi}\in Y\cap (\ker M_Z^\dagger)^\bot =Y\cap Z$. By \thm{szegedy}, the $-1$-eigenspace of $W$ is exactly $(Y\cap Z^\bot)\oplus (Y^\bot\cap Z)$ and the $1$-eigenspace of $W$ is exactly $(Y\cap Z)\oplus (Y^\bot\cap Z^\bot)$. 
\end{proof}

\begin{lemma}
$M_Y$, $R_Z=2\Pi_Z-I$ and $R_Y=2\Pi_Y-I$ can be implemented in time $O(\log n)$.
\end{lemma}
\begin{proof}
To implement $R_Z$ and $R_Y$, we need only show how to implement the unitary versions of $M_Z$ and $M_Y$. We begin with $M_Z$. For any $u\in [n]$, we can map $\ket{u}\mapsto \ket{0,u,u,0}$ by initializing three new registers and copying $u$ into one of them. Then we map:
$$\ket{0,u,u,0}\mapsto \ket{0,u,u}\frac{1}{\sqrt{n-1}}\sum_{v\neq u}\ket{v}\overset{H\otimes I^{\otimes 3}}{\mapsto} \frac{1}{\sqrt{2(n-1)}}\(\ket{0,u,u}\sum_{v\neq u}\ket{v}+\ket{1,u,u}\sum_{v\neq u}\ket{v}\)\mapsto \ket{x_u},$$
where the last transformation is achieved by swapping the last two registers conditioned on the first. This can be implemented in $O(\log n)$ elementary gates. 

For $M_Y$, we start by mapping any edge $\ket{u,v}$ to $\ket{1,0,u,v}$, followed by:
$$\ket{1,0,u,v}\overset{H\otimes I^{\otimes 3}}{\mapsto} \frac{1}{\sqrt{2}}\(\ket{0,0,u,v}-\ket{1,0,u,v}\)\mapsto \frac{1}{\sqrt{2}}\(\ket{0,u,u,v}-\ket{1,v,u,v}\)=\ket{y_{u,v}},$$
where in the last step we copy either $u$ or $v$ into the second register depending on the value of the first register. This can be implemented in $O(1)$ elementary gates. 

Then in order to implement $R_Z$, we simply apply ${M}_Z^\dagger$, reflect about $\mathrm{span}\{\ket{0,u,u,0}:u\in [n]\}$, and then apply ${M}_Z$ again. To implement $R_Y$, we apply ${M}_Y^\dagger$, reflect about $\mathrm{span}\{\ket{1,0,u,v}:u,v\in [n],u\neq v\}$, and then apply ${M}_Y$. 
\end{proof}

We now show how to efficiently implement the span program $P^{\beta}$ when $1/n^{O(1)}\leq \beta \leq n^{O(1)}$. First, consider $\ket{w_0}$, the minimal positive witness for $P$.
Since $\ket{w_0}$ corresponds to an optimal $st$-flow in the complete graph, it is easy to compute that 
$$\ket{w_0}=\frac{1}{n}\ket{s,t}+\frac{1}{2n}\sum_{u\in [n]\setminus \{s,t\}}(\ket{s,u}+\ket{u,t}) - \frac{1}{n}\ket{t,s}-\frac{1}{2n}\sum_{u\in [n]}(\ket{t,u}+\ket{u,s}),$$
and $\norm{\ket{w_0}}^2=\frac{1}{n}$ (see also \lem{w0}).
We can construct this state by mapping $\ket{s,0}+\ket{0,t}\mapsto \sum_{u\neq s}\ket{s,u}+\sum_{u\neq t}\ket{u,t}$ and then performing a swap controlled on an additional register in the state $\frac{1}{\sqrt{2}}(\ket{0}+\ket{1})$. 
The initial state of the scaled span program $P^{\beta}$ is (see \thm{scaling}):
$$\ket{w_0^\beta}=\frac{\beta}{\beta^2+\frac{1}{n}}\ket{w_0}+\frac{\frac{1}{n}}{\beta^2+\frac{1}{n}}\ket{\hat 0}+\frac{\beta}{\sqrt{\beta^2+\frac{1}{n}}}\ket{\hat 1},$$
which we can also construct efficiently, as follows:
$$\ket{\hat{0}}\mapsto \frac{\beta\sqrt{n}}{\beta^2+\frac{1}{n}}\ket{\hat 2}+\frac{1}{{n}\beta^2+1}\ket{\hat 0}+\frac{\beta}{\sqrt{\beta^2+\frac{1}{n}}}\ket{\hat 1}\mapsto \frac{\beta}{\beta^2+\frac{1}{n}}\ket{w_0}+\frac{\frac{1}{n}}{\beta^2+\frac{1}{n}}\ket{\hat 0}+\frac{\beta}{\sqrt{\beta^2+\frac{1}{n}}}\ket{\hat 1}.$$
The first step is accomplished by a pair of rotations using $O(\log \frac{n}{\beta})$ elementary gates, and the second is accomplished by mapping $\ket{\hat{2}}$ to $\frac{\ket{w_0}}{\norm{\ket{w_0}}}=\sqrt{n}\ket{w_0}$, which can be accomplished in $O(\log n)$ elementary gates. 

Next, we have 
$A^{\beta}=\beta A+(\ket{s}-\ket{t})\bra{\hat{0}}+\frac{\sqrt{\beta^2+\frac{n}{2}}}{\beta}\ket{\hat{1}}\bra{\hat{1}}$, so 
$$\ker A^{\beta}\oplus \mathrm{span}\{\ket{w_0^\beta}\}=\ker A\oplus \mathrm{span}\{\ket{\hat{0}}-\frac{1}{\beta}\ket{w_0}\}\oplus \mathrm{span}\{\ket{w_0^\beta}\}.$$
 We know how to reflect about $\ker A$, and since we can efficiently construct $\ket{w_0^\beta}$, we can reflect about it, so we need only consider how to reflect about $\span\{\ket{\hat 0}-\frac{1}{\beta}\ket{w_0}\}$. Since we can compute $\ket{w_0}$ efficiently, we can compute:
$$\ket{\hat{0}}\mapsto \frac{\beta}{\sqrt{\beta^2+1}}\ket{\hat 0}+\frac{1}{\sqrt{\beta^2+1}}\ket{\hat 1}\mapsto \frac{\beta}{\sqrt{\beta^2+1}}\ket{\hat 0}+\frac{1}{\sqrt{\beta^2+1}}\ket{\bar{w}_0}.$$
The first step is a rotation, which can be performed in $O(\log \frac{1}{\beta})$ elementary gates, and the second step is some mapping that maps $\ket{\hat 1}$ to $\ket{w_0}$, which we know can be done in $O(\log n)$ elementary gates. Thus, the total cost to reflect about $\ker A^{\beta}$ is $O(\log n)$. 
\end{proof}

\end{document}